\documentclass[twoside]{article}
\PassOptionsToPackage{compress}{natbib}

\usepackage[accepted]{aistats2023}

\usepackage[utf8]{inputenc} 
\usepackage[T1]{fontenc}    
\usepackage{hyperref}       
\usepackage{url}            
\usepackage{booktabs}       
\usepackage{amsmath,amsfonts}       
\usepackage{nicefrac}       
\usepackage{microtype}      
\usepackage{xcolor}         
\usepackage{adjustbox}
\usepackage{diagbox}
\usepackage{wrapfig,lipsum,booktabs}
\usepackage{arydshln}
\usepackage{subcaption}
\usepackage{lineno}
\usepackage{natbib}
\setcitestyle{authoryear,open={(},close={)}}

\usepackage{amsmath,amsthm,amsfonts,amssymb}
\usepackage{bm}
\usepackage{graphicx}
\usepackage{sidecap}
\usepackage{mathrsfs}
\usepackage{multirow, multicol}
\usepackage{caption}

\usepackage{wrapfig}
\usepackage{booktabs}

\usepackage{algorithm}
\usepackage{algorithmic}
\usepackage{xcolor}
\usepackage{amsmath,amsfonts,bm}
\usepackage{xspace}
\newtheorem{theorem}{Theorem}[section]
\newtheorem{definition}{Definition}[section]
\newtheorem{lemma}[theorem]{Lemma}
\newtheorem{corollary}{Corollary}[theorem]

\newcommand{\MI}{\mathcal{I}}

\newcommand{\TC}{\mathcal{TC}}
\newcommand{\TCeld}{\widehat{\mathcal{TC}}}

\newcommand{\Enpy}{\text{H}}

\newcommand{\bbR}{\mathbb{R}}

\newcommand{\bbE}{\mathbb{E}}

\newcommand{\bbN}{\mathbb{N}}

\newcommand{\bbP}{\mathbb{P}}
\newcommand{\bbQ}{\mathbb{Q}}

\newcommand{\calA}{\mathcal{A}}
\newcommand{\calB}{\mathcal{B}}

\newcommand{\calD}{\mathcal{D}}
\newcommand{\calE}{\mathcal{E}}
\newcommand{\calF}{\mathcal{F}}

\newcommand{\calH}{\mathcal{H}}

\newcommand{\calL}{\mathcal{L}}

\newcommand{\calN}{\mathcal{N}}

\def\Figref#1{Figure~\ref{#1}}

\def\Secref#1{Section~\ref{#1}}

\def\1{\bm{1}}

\def\vu{{\bm{u}}}
\def\vv{{\bm{v}}}

\def\vx{{\bm{x}}}
\def\vy{{\bm{y}}}
\def\vz{{\bm{z}}}

\def\vX{{\bm{X}}}

\def\mX{{\bm{X}}}

\DeclareMathAlphabet{\mathsfit}{\encodingdefault}{\sfdefault}{m}{sl}
\SetMathAlphabet{\mathsfit}{bold}{\encodingdefault}{\sfdefault}{bx}{n}

\def\sI{{\mathbb{I}}}

\newcommand{\E}{\mathbb{E}}

\def\Figref#1{Figure~\ref{#1}}
\def\Tableref#1{Table~\ref{#1}}

\def\Secref#1{Section~\ref{#1}}

\def\eqref#1{equation~\ref{#1}}

\usepackage{amssymb}
\usepackage{pifont}

\def\dis{{\mathcal{H}}}
\def\det{\text{Det}}

\begin{document}

\twocolumn[

\aistatstitle{Estimating Total Correlation with Mutual Information Estimators}


\aistatsauthor{ Ke Bai$^*$ \And Pengyu Cheng$^*$ \And  Weituo Hao \And  Ricardo Henao \And Lawrence Carin}

\aistatsaddress{ Duke University \And  Tencent AI Lab \And ByteDance  \And KAUST \And  KAUST } ]

\begin{abstract}
\vspace{-3mm}
Total correlation (TC) is a fundamental concept in information theory which measures statistical dependency among multiple random variables. Recently, TC has shown noticeable effectiveness as a regularizer in many learning tasks, where the correlation among multiple latent embeddings requires to be jointly minimized or maximized. However, calculating precise TC values is challenging, especially when the closed-form distributions of embedding variables are unknown. In this paper, we introduce a unified framework to estimate total correlation values with sample-based mutual information (MI) estimators. More specifically, we 
discover a relation between TC and MI and propose two types of calculation paths (tree-like and line-like) to decompose TC into MI terms. With each MI term being bounded, the TC values can be successfully estimated. Further, we provide theoretical analyses concerning the statistical consistency of the proposed TC estimators. Experiments are presented on both synthetic and real-world scenarios, where our estimators demonstrate effectiveness in all TC estimation, minimization, and maximization tasks.  The code is available at \textit{https://github.com/Linear95/TC-estimation}.
\end{abstract}

\vspace{-5mm}
\section{INTRODUCTION}
\vspace{-2.mm}
Statistical dependency measures the association (correlation) of variables or factors in models and systems, and constitutes one of the key considerations in various scientific domains including  statistics~\citep{granger1994using,jiang2015nonparametric},  robotics~\citep{julian2014mutual,charrow2015information}, bioinformatics~\citep{lachmann2016aracne,zea2016mitos}, and machine learning~\citep{chen2016infogan,alemi2016deep,hjelm2018learning}. In deep learning, statistical dependency has been applied as learning objective or regularizer in many well-known training frameworks, such as information bottleneck~\citep{alemi2016deep}, disentangled representation learning~\citep{chen2018isolating,pmlr-v97-peng19b,cheng2022replacing} and contrastive learning~\citep{chen2020simple,gao2021simcse}. Recent neural network studies have also demonstrated the benefits of considering statistical dependency in terms of model robustness~\citep{zhu2020learning}, fairness~\citep{pmlr-v97-creager19a}, interpretability~\citep{chen2016infogan,cheng2020improving},~\textit{etc}.

Among diverse measurement approaches for statistical dependency, the concept of mutual information (MI) is one of the most commonly used, especially in deep model training~\citep{alemi2016deep,belghazi2018mutual,chen2020simple}. Given two random variables $\vx$ and $\vy$ with joint and marginal distributions $p(\vx, \vy)$, $p(\vx)$ and $p(\vy)$, respectively, their mutual information is defined as:
\begin{equation}\label{eq:mi-definition} 
    \MI(\vx; \vy) = \bbE_{p(\vx, \vy)} \Big[\log \frac{p(\vx, \vy)}{p(\vx) p(\vy)} \Big].
\end{equation}
Recently, mutual information has been used as a training criterion to deliver noticeable performance improvement for deep models on various learning tasks such as conditional generation~\citep{chen2016infogan,cheng2020improving}, domain adaptation~\citep{gholami2020unsupervised,cheng2020club}, representation learning~\citep{chen2020simple,gao2021simcse,yuan2020improving}, and model debiasing~\citep{song2019learning,cheng2020fairfil}.
However, standard MI in \eqref{eq:mi-definition} only handles the statistical dependency between two variables. When considering correlation among multiple variables, MI requires computation between each pair of variable, which leads to a quadratic scaling in computation cost. To address this problem, total correlation (TC)~\citep{watanabe1960information} or multi-information~\citep{studeny1998multiinformation} has been proposed for  multi-variable scenarios:
\begin{align}
    \TC(\mX) = &\TC(\vx_1,\vx_2,\dots,\vx_n) \label{eq:tc-definition} \\
     =&  \bbE_{p(\vx_1, \vx_2, \dots, \vx_n)} \Big[\log \frac{p(\vx_1, \vx_2, \dots, \vx_n)}{p(\vx_1) p(\vx_2)\dots p(\vx_n) }\Big].\nonumber
\end{align}
TC has been also proven effective to enhance machine learning models in many tasks such as independent component analysis~\citep{cardoso2003dependence}, structure discovery~\citep{ver2014discovering} and disentangled representation learning~\citep{chen2018isolating,locatello2019fairness,kim2018disentangling}. However, TC suffers from the same practical problem as MI, namely, that the exact values of TC are difficult to calculate without the availability of the closed-form distributions, or only relying on samples either via the kernel density estimation (KDE)~\citep{kandasamy2015nonparametric, singh2014exponential} and \textit{k}-Nearest Neighbor (\textit{k}-NN)~\citep{pal2010estimation, Kraskov2004EstimatingMI, gao2018demystifying}, which may perform well on data with low dimensionality but failed on the high one. Furthermore, previous works on disentangled representation learning~\citep{chen2018isolating,gao2019auto} strongly assume that both the joint distribution $p(\vx_1,\vx_2,\dots,\vx_n)$ and marginal distributions $\{p(\vx_i)\}_{i=1}^n$ follow Gaussian distributions, so that the TC value can be calculated in closed-form. \citet{poole2019variational} proposed an upper bound of  TC by further introducing an auxiliary variable $\vy$. With a strong assumption that given $\vy$ and $\vx_i|\vy$, all variables are conditionally independent, $p(\mX|\vy) = \prod_{i=1}^n p(\vx_i|\vy)$,
\citet{poole2019variational} showed that $\TC(\mX) = \sum_{i=1}^n \MI(\vx_i; \vy) - \MI(\mX; \vy)$, \emph{i.e.}, that the TC value can be bounded by MI estimators~\citep{belghazi2018mutual,poole2019variational,cheng2020club}. All these TC estimation methods require additional strong assumptions on the distributions of data samples, which heavily limits their application in practice.

In this paper, we propose a new  \textbf{T}otal \textbf{C}orrelation \textbf{E}stimator via \textbf{L}inear \textbf{D}ecomposition (TCeld), which importantly, does not require any assumptions about sample distributions. More specifically, we discover a linear decomposition relation between TC and mutual information (MI). Based on this relationship, we linearly split the TC into several MI terms, then estimate each MI term aided with variational MI estimators. With two different TC decomposition paths (\textit{line-like} and \textit{tree-like}), we obtain two types of TC estimators, $\widehat{\TC}_\text{Tree}$ and $\widehat{\TC}_{\text{Line}}$, respectively. Moreover, we prove that the nice properties of MI estimators, such as consistency, are maintained in the corresponding TC estimators, thanks to the linearity of the proposed TC decomposition.
%
%
In the experiments, we first test the estimation quality of our TC estimators on simulation data sampled from multi-variate Gaussian distributions, then apply the TC estimators to two real-world TC optimization tasks. The numerical results demonstrate the effectiveness of the proposed TC estimators under both TC estimation and optimization scenarios. 

\vspace{-3mm}
\section{BACKGROUND}
\label{sec:background}
\vspace{-2mm}
\subsection{Sample-based Mutual Information Estimators}\label{sec:background-mi}
\vspace{-2mm}
Although mutual information (MI) is a fundamental tool for measuring statistical dependency between two variables, calculating MI values with only samples provided is challenging, especially when the closed-form distributions of variables are unknown.  To estimate MI values with samples, several variational MI estimators have been introduced. \citet{barber2003algorithm} approximate the conditional distribution $p(\vx|\vy)$ between $\vx$ and $\vy$ by a variational distribution $q(\vx|\vy)$, and derive:
\vspace{-2mm}
\begin{equation}\label{eq:ba-boud}
    \MI_{\text{BA}} := \Enpy (\vx) + \bbE_{p(\vx,\vy)}[\log q(\vx|\vy)],
\end{equation}
with $\Enpy(\vx)$ as the entropy of $\vx$.
 Utilizing the Donsker-Varadhan representation~\citep{donsker1983asymptotic} of KL-divergence, \citet{belghazi2018mutual} obtain an MI Neural Estimator (MINE) with a score network $\phi(\cdot, \cdot)$:
%
%
\begin{align}
    \MI_{\text{MINE}} :=& \bbE_{p(\vx,\vy)}[\phi(\vx,\vy)] \label{eq:mine} 
    \\ &- \log(\bbE_{p(\vx)p(\vy)}[\exp({\phi(\vx,\vy)})]. \nonumber
\end{align}
Nguyen, Wainwright, and Jordan
(NWJ)  derive another lower bound considering MI as in a f-divergence form~\cite{nguyen2010estimating}, which also requires a score network $\phi(\cdot, \cdot)$:
\begin{equation}
  \MI_{\text{NWJ}}:=\bbE_{p(\vx,\vy)} [\phi(\vx,\vy)] - \bbE_{p(\vx)p(\vy)} [ e^{\phi(\vx, \vy)-1}]. \label{eq:NWJ}
\end{equation}
With Noise Contrastive Estimation (NCE)~\citep{gutmann2010noise}, \citet{oord2018representation} propose a MI lower bound called InfoNCE, based on a group of  samples $\{(\vx_i, \vy_i)\}_{i=1}^N$, to obtain a low-variance estimator:
\begin{equation}
    \MI_\text{InfoNCE} := \bbE [\frac{1}{N} \sum_{i=1}^N \log \frac{\exp({\phi(\vx_i, \vy_i)})}{\frac{1}{N}\sum_{j=1}^N \exp({\phi(\vx_i, \vy_j)})}] .\label{eq:infonce}
\end{equation}
Different from the MI lower bounds introduced above, \citet{cheng2020club} derive a Contrastive Log-ratio Upper Bound (CLUB), which is also based on a variational approximation $q(\vx|\vy)$ of $p(\vx|\vy)$:
\begin{equation}
    \MI_{\text{CLUB}} := \bbE_{p(\vx,\vy)}[\log q(\vx|\vy)] \label{eq:club}- \bbE_{p(\vx)p(\vy)}[\log q(\vx|\vy)]
\end{equation}
%
%

\vspace{-3mm}
\subsection{Statistical Properties of Estimators}\label{sec:background-prop}
\vspace{-2mm}
Given observed samples $\vx^{1}, \vx^{2}, \dots \vx^{m}  \sim p(\vx)$, a statistic is defined as $T(\vx^{1}, \vx^{2}, \dots, \vx^{m})$, where $T(\cdot)$ is an arbitrary real-valued function taking samples $\vx^{1}, \vx^{2}, \dots \vx^{m}$ as inputs~\citep{degroot2012probability}.  Sample-based estimators to calculate statistical dependency, \textit{i.e.} mutual information and total correlation, are also examples of statistics.  To evaluate the performance of a statistic,  \textit{statistical properties} are introduced to describe behaviors of the statistic under different data situations~\citep{degroot2012probability} (\textit{e.g.}, with large/small sample number, with/without outlier).  
For estimation of mutual information and total correlation,  we follow prior works~\citep{paninski2003estimation,belghazi2018mutual} and will mainly consider the following three key properties:
\begin{definition}[Unbiasedness]\label{def:unbiasedness}
An estimator $\hat{R} = T(\vx^{1}, \vx^{2}, \dots, \vx^{m})$ is unbiased of the ground true value $R_p$ with respect to the distribution $p(\vx)$, if $\bbE_{p(\vx)}[\hat{R}] = R_p$.
\end{definition}
\begin{definition}[Consistency]\label{def:consistency}
An estimator $\hat{R}_m = T(\vx^{1}, \vx^{2}, \dots, \vx^{m})$ is consistent to the ground true value $R_p$ with respect to the distribution $p(\vx)$, if  $\forall \varepsilon>0$,  $\lim_{m \rightarrow \infty} \bbP\{| \hat{R}_m - R_p | \geq \varepsilon \} = 0$.
\end{definition}
\begin{definition}[Strong Consistency]\label{def:strong-consistency}
An estimator $\hat{R}_m = T(\vx^{1}, \vx^{2}, \dots, \vx^{m})$ is strongly consistent to the ground true value $R_p$ with respect to the distribution $p(\vx)$, if $\forall \varepsilon>0$, there exists an integer $M>0$ such that $\forall m > M$, $| \hat{R}_m - R_p | \leq \varepsilon$ almost surely. 
\end{definition}

\vspace{-2mm}
\section{TOTAL CORRELATION ESTIMATION} 
\vspace{-2mm}
Suppose $m$ groups of data samples $\{\mX^{j}\}_{j=1}^m = \{(\vx^{j}_1, \vx^{j}_2, \dots, \vx^{j}_n)\}_{j=1}^m$ are observed from an unknown distribution $p(\mX) = p(\vx_1, \vx_2, \dots, \vx_n)$. We seek to estimate the total correlation (TC) of these $n$ variables as in \eqref{eq:tc-definition}. Below we first describe the proposed sample-based TC estimators, then analyze their statistical properties.
\vspace{-1mm}
\subsection{Sample-based TC Estimators}
\vspace{-1mm}
With the definition of total correlation (TC) and mutual information (MI) in \eqref{eq:tc-definition} and \eqref{eq:mi-definition}, we discover a connection between TC and MI and summarize it in the following Theorem~\ref{thm:general_connection}.
%
\begin{theorem}\label{thm:general_connection}
Let $\mX=(\vx_1,\vx_2,\dots,\vx_n)$ be a group of random variables. Suppose set $\calA = \{i_1, i_2, \dots, i_k\} \subseteq \{1,2,\dots,n\}$ is an index subgroup.  $\bar{\calA}= \{i:i\notin\calA\}$ is the complementary set of $\calA$. Denote $\mX_{\calA} = (\vx_{i_1}, \vx_{i_2}, \dots, \vx_{i_k})$ as the selected variables from $\mX$ with the indexes $\calA$. Then we have $\TC(\mX) = \TC(\mX_{\calA}) + \TC(\mX_{\bar{\calA}}) + \MI(\mX_{\calA}; \mX_{\bar{\calA}})$.
\end{theorem}
Theorem~\ref{thm:general_connection} underscores the insight that the TC of a group of variables $\mX$ can be decomposed into the TC of two subgroups $\TC(\mX_{\calA})$ and $\TC(\mX_{\bar{\calA}})$ and the MI between the two subgroups $\MI(\mX_{\calA}; \mX_{\bar{\calA}})$. Therefore, we can recursively represent the TC of the subgroups in terms of MI terms for lower-level subgroups. With this decomposition strategy, we can effectively cast $\TC(\mX)$ into a summation of MI terms between variable subgroups. 
Note that the decomposition form depends on the separation of variables in each subgroup. In the following, we introduce two types of variable separation strategies: \textit{line-like} and \textit{tree-like} decomposition as shown in Figure~\ref{fig:calculation_path}.

\begin{figure*}[t]
    \centering
    \includegraphics[width=0.9\textwidth]{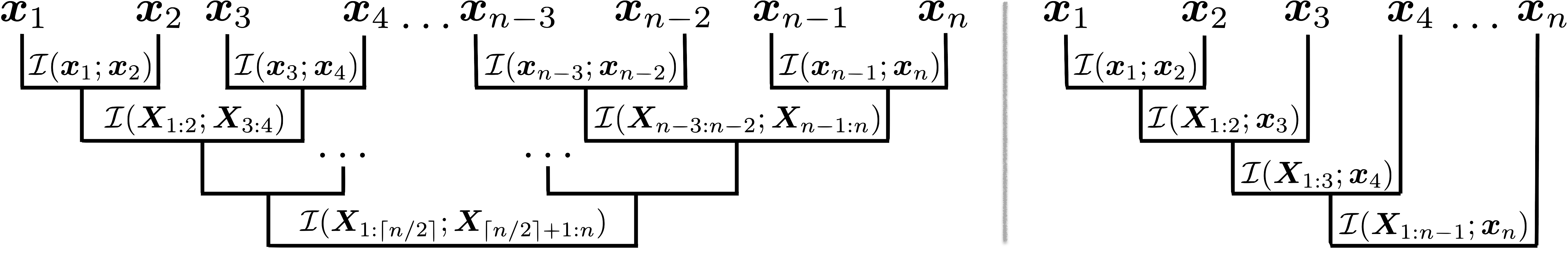}
    \vspace{-2mm}
    \caption{Two decomposition paths of total correlation $\TC(\vx_1,\vx_2,\dots,\vx_n)$. \textbf{Left} (tree-like decomposition): Divide the variables in a group into two subgroups with similar sizes. Calculate the MI between the subgroups and recursively calculate the TC of both subgroups. $\lceil n/2\rceil$ is the smallest number larger than $n/2$. \textbf{Right} (line-like decomposition): Calculate the MI between the current group of variables and the next variable, and then add the next variable into current group.}
    \label{fig:calculation_path}
     \vspace{-2mm}
\end{figure*}

\textbf{Line-like Decomposition } 
In each variable subgroup separation, our line-like decomposition strategy splits out a single variable.
Let $\mX_{i:j} = (\vx_i, \vx_{i+1}, \dots, \vx_{j})$ denote  a subset of variables with indexes from $i$ to $j$. Then we can extend  Theorem~\ref{thm:general_connection} to the following Corollary~\ref{thm:one-variable-split} and Corollary~\ref{thm:line-like}. 
Based on Corollary~\ref{thm:one-variable-split}, the line-like decomposition can be dynamically described as:
\begin{equation}\label{eq:line-like}
    \TC(\mX_{1:i+1}) = \TC(\mX_{1:i}) + \MI(\mX_{1:i}; \vx_{i+1}). 
\end{equation}
Iteratively applying \eqref{eq:line-like} to the remaining TC term, we derive the representation of $\TC(\mX)$ as the summation of MI terms in Corollary~\ref{thm:line-like}. With a given MI estimation method $\hat{\MI}$ applied to each MI term in Corollary~\ref{thm:line-like}, our line-like TC estimator can be calculated as:
\begin{equation}
    \widehat{\TC}_\text{Line}[\hat{\MI}](\mX) = \sum_{i=1}^{n-1} \hat{\MI}(\mX_{1:i} ; \vx_{i+1}) .
\end{equation}
\begin{corollary}\label{thm:one-variable-split}
Given the group $\mX$ and another variable  $\vy$, $ \TC(\mX \cup \{\vy\}) = \TC(\mX) + \MI(\mX; \vy).$
\end{corollary}
\begin{corollary}\label{thm:line-like}
Given $\mX = (\vx_1, \vx_2, \dots, \vx_n)$, we have $\TC(\mX) =  \sum_{i=1}^{n-1} \MI(\mX_{1:i} ; \vx_{i+1}).$
\end{corollary}
\textbf{Tree-like Decomposition } In each variable subgroup separation, the tree-like decomposition strategy separates variables $\mX_{i:j}$ into balanced variable subgroups with similar sizes in the following way:
\begin{align}
\TC(\mX_{i:j}) =& \TC(\mX_{i:\lfloor {(i+j)}/{2} \rfloor}) + \TC(\mX_{\lfloor {(i+j)}/{2} \rfloor+1:j}) \nonumber \\ &+ \MI(\mX_{i:\lfloor {(i+j)}/{2} \rfloor};\mX_{\lfloor {(i+j)}/{2} \rfloor+1:j}),  
\end{align}
where $\lfloor t\rfloor$ indicates the largest integer smaller than $t$. In accordance with the line-like decomposition, iteratively applying this dichotomous dynamic will finally convert $\TC(\mX)$ into the summation of MI terms. Since the closed-form of this tree-like TC estimator is hard to summarize in an equation, we describe it recursively in Algorithm~\ref{alg:tc-tree-like-decompostion}.

We call this novel TC estimator as  \textbf{T}otal \textbf{C}orrelation \textbf{E}stimation with \textbf{L}inear \textbf{D}ecomposition (TCeld). From the linearity of the above two decomposition strategies, one can easily derive:
\begin{theorem}\label{thm:TC-MI-upper-lower-bounds}
 If an MI estimator $\hat{\MI}$ is an MI upper (or lower) bound, then the corresponding $\widehat{\TC}_\text{Line}[\hat{\MI}]$ and $\widehat{\TC}_\text{Tree}[\hat{\MI}]$ are both the TC upper (or lower) bounds.
\end{theorem}
Therefore, by selecting an MI lower bound $\hat{\MI}_0$ and an MI upper bound $\hat{\MI}_1$, we can limit the ground-truth TC value in a certain range, $\widehat{\TC}^*[\hat{\MI_0}](\mX) \leq \TC(\mX) \leq \widehat{\TC}^*[\hat{\MI_1}](\mX)$, where $\widehat{\TC}^*$ can be either line-like $\widehat{\TC}_\text{Line}$ or tree-like $\widehat{\TC}_\text{Tree}$.

Both the line-like and tree-like TC estimators have no additional requirement on the selection of MI estimators. However, the statistical performance of the proposed TC estimators highly depends on the selected MI estimators. To further analyze the influence of MI bounds choice for TC estimators, we discuss the statistical properties between the TC estimators and MI estimators in the next subsection.    


%

\begin{algorithm}[t]
\begin{algorithmic}
\STATE \textbf{Prerequisite:} MI estimation method $\hat{\MI}$, samples $\{\mX_{1:n}^{j} \}_{j=1}^m =\{(\vx^{j}_1, \vx^{j}_2, \dots, \vx^{j}_n )\}_{j=1}^m$
 \STATE \textbf{Function} $\widehat{\TC}_\text{Tree}[\hat{\MI}](\mX_{i:j})$\textbf{:}
 \IF {$j-i < 1$}
     \RETURN 0
 \ELSE
    \STATE $m =\lfloor {(i+j)}/{2} \rfloor $
    \RETURN $\widehat{\TC}_\text{Tree}[\hat{\MI}](\mX_{i:m}) +  \widehat{\TC}_\text{Tree}[\hat{\MI}](\mX_{m+1:j}) + \hat{\MI}(\mX_{i:m};\mX_{m+1:j})$
 \ENDIF
\end{algorithmic}
 \caption{Tree-like TC estimator $\widehat{\TC}_\text{Tree}[\hat{\MI}]$ calculation.} \label{alg:tc-tree-like-decompostion}
\end{algorithm}
 
 \vspace{-1.5mm}
\subsection{Statistical Properties of TC Estimators}
\vspace{-1.5mm}
As introduced in Section~\ref{sec:background-prop}, TC estimators can be also regarded as a type of statistics based on observed groups of sample data. Therefore, the aforementioned statistical properties (unbiasedness, consistency, and strong consistency in Section~\ref{sec:background-prop}) are applicable to TC estimators.
Thanks to the form of linear combinations of MI terms in our TC estimators, we find the following relations between TC and MI estimators in terms of statistical properties: 
\begin{theorem}\label{thm:linear-prop-unbiased}
 If an MI estimator $\hat{\MI}$ is unbiased, then the corresponding TC estimators $\widehat{\TC}_{\text{Line}}[\hat{\MI}]$ and $\widehat{\TC}_{\text{Tree}}[\hat{\MI}]$ are both unbiased.
\end{theorem}
\begin{theorem}\label{thm:linear-prop-consistent}
 If an MI estimator $\hat{\MI}$ is (strongly) consistent, then the corresponding TC estimators $\widehat{\TC}_{\text{Line}}[\hat{\MI}]$ and $\widehat{\TC}_{\text{Tree}}[\hat{\MI}]$ are both (strongly) consistent.
\end{theorem}
The definitions of \textit{unbiasedness} and \textit{consistency} are introduced in Section~\ref{sec:background-prop}. The details of the proofs for both theorems are shown in the Supplementary Material. These two Theorems indicate that the unbiasedness and the consistency of the MI estimators is conveniently inherited by the corresponding TC estimators. Note that \citet{belghazi2018mutual} have shown that  MINE (in \eqref{eq:mine}) MI estimator is strongly consistent. Consequently, we have: 
n\begin{corollary}\label{thm:TC-mine-consistent}
There exists an score network function family $\{\phi(\cdot, \cdot)\}$ (as described in \eqref{eq:mine}), such that  $\widehat{\TC}_{\text{Line}}[\hat{\MI}_{\text{MINE}}]$ and $\widehat{\TC}_{\text{Tree}}[\hat{\MI}_{\text{MINE}}]$ are strongly consistent.
\end{corollary}
Apart from the MINE MI estimator, we also analyze the asymptotic behaviors of other MI variational estimators (InfoNCE in \eqref{eq:infonce}, NWJ in \eqref{eq:NWJ}, and CLUB in \eqref{eq:club}). We found that both InfoNCE and NWJ are strongly consistent while CLUB does not guarantee the consistency (supportive proofs are provided in Supplementary Material).
Therefore, we summarize the consistency of the corresponding TC estimators in the corollaries below:
\begin{corollary}\label{thm:TC-infonce-and-nwj-consistent}
For any MI estimator $\hat{\MI} \in \{\hat{\MI}_\text{InfoNCE}, \hat{\MI}_\text{NWJ} \}$, there exists a score network function family $\{\phi(\cdot, \cdot)\}$, such that  $\widehat{\TC}_{\text{Line}}[\hat{\MI}]$ and $\widehat{\TC}_{\text{Tree}}[\hat{\MI}]$ are both strongly consistent.
\end{corollary}
Moreover, strong consistency is a sufficient condition for consistency~\citep{degroot2012probability}, hence all of TC-MINE, TC-NWJ and TC-InfoNCE estimators are also consistent statistics.
For unbiasedness, though Theorem~\ref{thm:linear-prop-unbiased} indicates that unbiased MI estimators can lead to unbiased TC estimators, to the best of our knowledge, none of the previous variational MI estimators~\citep{belghazi2018mutual,oord2018representation,poole2019variational,cheng2020club} are unbiased. Therefore, we leave the study of unbiased TC estimators for future work. Besides the above theoretic analysis, we empirically test our TC estimator in Section~\ref{sec:experiments} with the application tasks introduced in Section~\ref{sec:related-work}.

\begin{figure*}[t]
    \centering
    \includegraphics[width=0.9\textwidth]{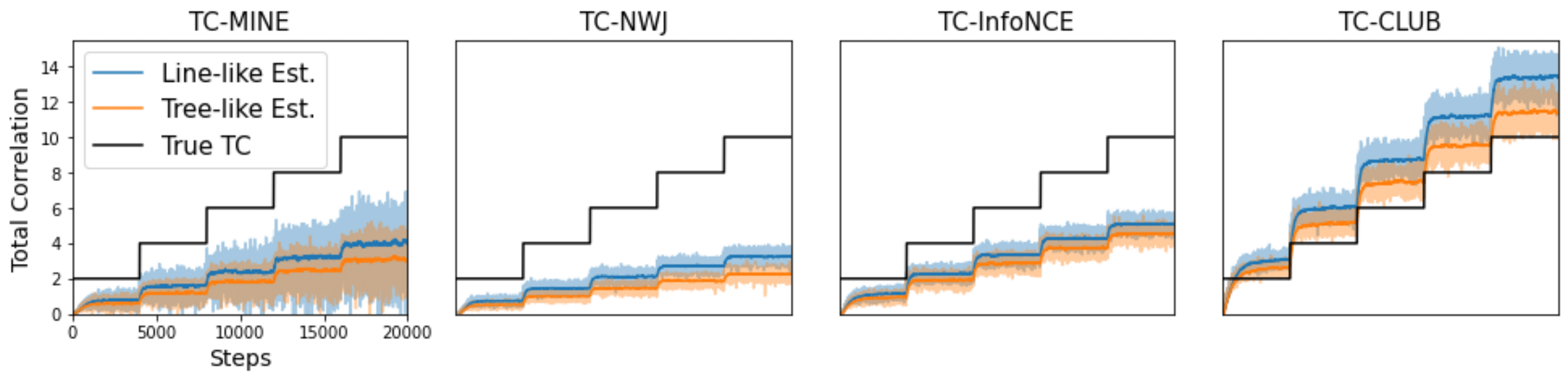}
    \vspace{-2mm}
     \caption{Simulation performance of TC \textit{line-like} and \textit{tree-like}  estimators with different MI estimators.}  \label{fig:TC-est-results}
    \vspace{-2.5mm}
\end{figure*}

\vspace{-2mm}
\section{RELATED WORK}
\label{sec:related-work}
\vspace{-2mm}
\textbf{Disentangled Representation Learning }
Disentangled representation learning (DRL) seeks to map each  data instance into independent latent subspaces, while different subspaces meaningfully  represent different attributes of the instance~\citep{locatello2019challenging}. 
Recently, DRL methods have attracted considerable interest on various learning tasks such as  domain adaptation~\citep{gholami2020unsupervised}, conditional generation~\citep{burgess2018understanding}, and few-shot learning~\citep{yuan2020improving}.
DRL methods are mainly recognized  into two categories as  unsupervised and supervised DRL. Prior unsupervised DRL works \citep{burgess2018understanding,kim2018disentangling} utilize different regularizers to make each dimension of latent space to be as independent as possible, which 
has been challenged by \citet{locatello2019challenging} in that each embedding dimension may not be related to a meaningful data attribute. Alternatively, supervised DRL methods~\citep{hjelm2018learning,kim2018disentangling,cheng2020improving,yuan2020improving} add different supervision terms on different embedding components, which effectively learn meaningful embedding while enabling disentanglement. 
Both supervised and unsupervised DRL methods require correlation reduction technique to prevent the embedding information from leaking into other embedding components. 
%
%
To reduce embedding correlation, \citet{hjelm2018learning,kim2018disentangling} use adversarial training methods, while \citet{chen2018isolating,cheng2020improving, yuan2020improving} minimize  statistical dependency (\textit{i.e.}, MI and TC), between different embedding components.

\textbf{Contrastive Representation Learning }
Contrastive representation learning is a fundamental training methodology which maximizes the difference of positive and negative data pairs to obtain informative representations. In contrastive learning, a pairwise distance/similarity score function is always set to measure data pairs. Then, the learning objective is to maximize the margin between scores of positive and negative data pairs. Prior contrastive learning has achieved satisfying performance in numerous tasks, such as metric learning~\citep{weinberger2006distance,davis2007information}, word embedding learning~\citep{mikolov2013efficient}, and graph embedding learning~\citep{tang2015line,grover2016node2vec}. 
Recently, contrastive learning has been recognized as a powerful tool in unsupervised or semi-supervised learning scenarios~\citep{he2020momentum,chen2020simple,gao2019auto}, which significantly narrows the gap of performance of supervised and unsupervised learning methods. Among these unsupervised methods, \citet{gao2021simcse} proposed a contrastive text representation learning framework (SimCSE). For each sentence, SimCSE use the dropout mechanism to generate sentence augmentation pairs, then maximize the mutual information within the augmented embedding pairs. The empirical results demonstrate contrastive learning obtains high-quality embeddings even with little supervision~\citep{gao2021simcse}.


\vspace{-2mm}
\section{EXPERIMENTS}
\label{sec:experiments}
\vspace{-2mm}
In this section, we empirically evaluate the effectiveness of our TCeld on three tasks: TC estimation, minimization, and maximization. For TC estimation, we generate synthetic data from correlated multi-variate Gaussian distributions, then compare the predictions from our TC estimator with the ground-truth TC values. For TC minimization, we conduct a multi-component disentangled representation learning experiment on Colored-MNIST~\citep{esser2020disentangling} dataset, to minimize the total correlation among the digit, style, and color embeddings of digit images. To test the TC maximization ability, we apply our TC estimator into a contrastive text learning framework to maximize the TC value among different sentence augmentations.  Since our proposed TC estimators can be flexibly induced by different MI estimators, for convenience, we refer the TC estimator as TC-(MI estimator name), or TC-(Line/Tree)-(MI estimator name) if the decomposition strategy is specified. For example, TC-Line-MINE denotes the TC estimator by line-like decomposition with the MINE MI estimator.

\vspace{-2mm}
\subsection{TC Estimation on Simulation Data}
\label{sec:exp-tc-estimation}
\vspace{-2mm}
We first test the estimation quality of our TC estimators under simulation scenarios.  We selected four MI estimators, MINE~\citep{belghazi2018mutual}, NWJ~\citep{nguyen2010estimating}, InfoNCE~\citep{oord2018representation}, and CLUB~\citep{cheng2020club}, to induce our TC estimators.  Then we test TCeld with both tree-like and line-like strategies. The detailed description and implementation of the four MI estimators are shown in the Supplementary Material, where the results of non-variational methods like KDE, \textit{k}-NN and their variants are also reported.

To evaluate TCeld's estimation performance with different ground-truth values, we sample simulated data from a four-dimensional Gaussian distributions $(\vx_1, \vx_2, \vx_3, \vx_4) \sim\calN(\bm{0}, \bm{\Sigma})$, where $\bm{\Sigma}$ is a covariance matrix with all diagonal elements equal to $1$, which means the variance of each $\vx_i$ is normalized to $1$. With this Gaussian assumption, the true TC value can be calculated in a closed-form as $\TC(\vx_1,\vx_2, \vx_3, \vx_4) = -\frac{1}{2} \log \text{Det}(\bm{\Sigma})$, where $\text{Det}(\bm{\Sigma})$ is the determinant of $\bm \Sigma$. Therefore, we can adjust the correlation coefficients (non-diagonal elements) in $\bm \Sigma$ to set the ground-truth TC values in the set $\{2.0,4.0,6.0,8.0,10.0\}$ (described in details in the Supplementary Material). The sample dimension is set to 20. The dimension of hidden states for variational estimators is 15. For each fixed TC value, we sample data 4000 times, with  batch size 64 and learning rate $0.001$ to train the estimators.


In Figures~\ref{fig:TC-est-results} we report the performance of TCeld with different MI bounds at each training step. In each figure, the true TC value is shown as a step function drawn as a black line. The line-like and tree-like estimation values are presented for different steps as shaded blue and orange curves respectively. The dark blue and orange curves illustrate the local averages of the estimated values, with a bandwidth equal to 200. With both  tree-like and line-like decomposition, the estimation values for TC-MINE, TC-NWJ and TC-InfoNCE stay below the truth TC step functions, supporting the claims in Theorem~\ref{thm:TC-MI-upper-lower-bounds}. For the only MI upper bound CLUB, the estimated values initially lie beneath the ground-truth TC, but finally converge above the step function. This is because at the beginning of the estimator training, the parameters are not well learned from the synthetic samples, and fail to support a valid MI upper bound. With the training progress going on, the estimator performs better and finally converges to the desired upper-bound estimation.


Furthermore, we provided the bias, variance, and the mean squared error (MSE) of TC estimation values in the 
 Supplementary Material.
The tree-like and line-like strategies have insignificant effects on variance of our TC estimators, where TC-NWJ always keep the lowest variance. However, as for bias and MSE, TC-CLUB works uniformly better than other estimators under line-like step, while TC-InfoNCE outperforms others mostly under tree-like decomposition. 
By further analyzing this phenomenon, we find that  when estimating $\MI(\vv;\vu)$, CLUB requires a variational approximation $q_\theta(\vv|\vu)$. When we use the line-like decomposition strategy, $\vv = \vx_{i+1} $ is always a single variable, and $\vu = \mX_{1:i}$ is the concatenation of $(\vx_1,\dots,\vx_i)$. The $q_\theta(\vv| \vu)$ with a neural network implementation can have better performance with output $\vv$ in a fixed low dimension. However, all the other MI estimators need to learn a score function $\phi(\vv,\vu)$, where the imbalanced inputs $\vv= \vx_{i+1}$ and $\vu= \mX_{1:i}$ can hinder the learning of function $\phi$.
In contrast, the tree-like TC estimators split variables equally into subgroups, which facilitate the learning of $\phi(\vu, \vv)$ with $\vu=\mX_{i:\lfloor {(i+j)}/{2} \rfloor}$ and $\vu=\mX_{\lfloor {(i+j)}/{2} \rfloor+1:j}$ for the lower-bound methods. For CLUB, the tree-like decomposition increases the output dimension of the variational net $q_\theta(\vv|\vu)$ and makes the learning more challenging, explaining TC-CLUB's lower performance than TC-InfoNCE in Supplementary Material.



%
\vspace{-2mm}
\subsection{TC Minimization for Disentangled Representation  Learning}\label{sec:exp-tc-minimization}
\vspace{-2mm}

\begin{figure}[t]
\centering
\includegraphics[width=0.68\columnwidth]{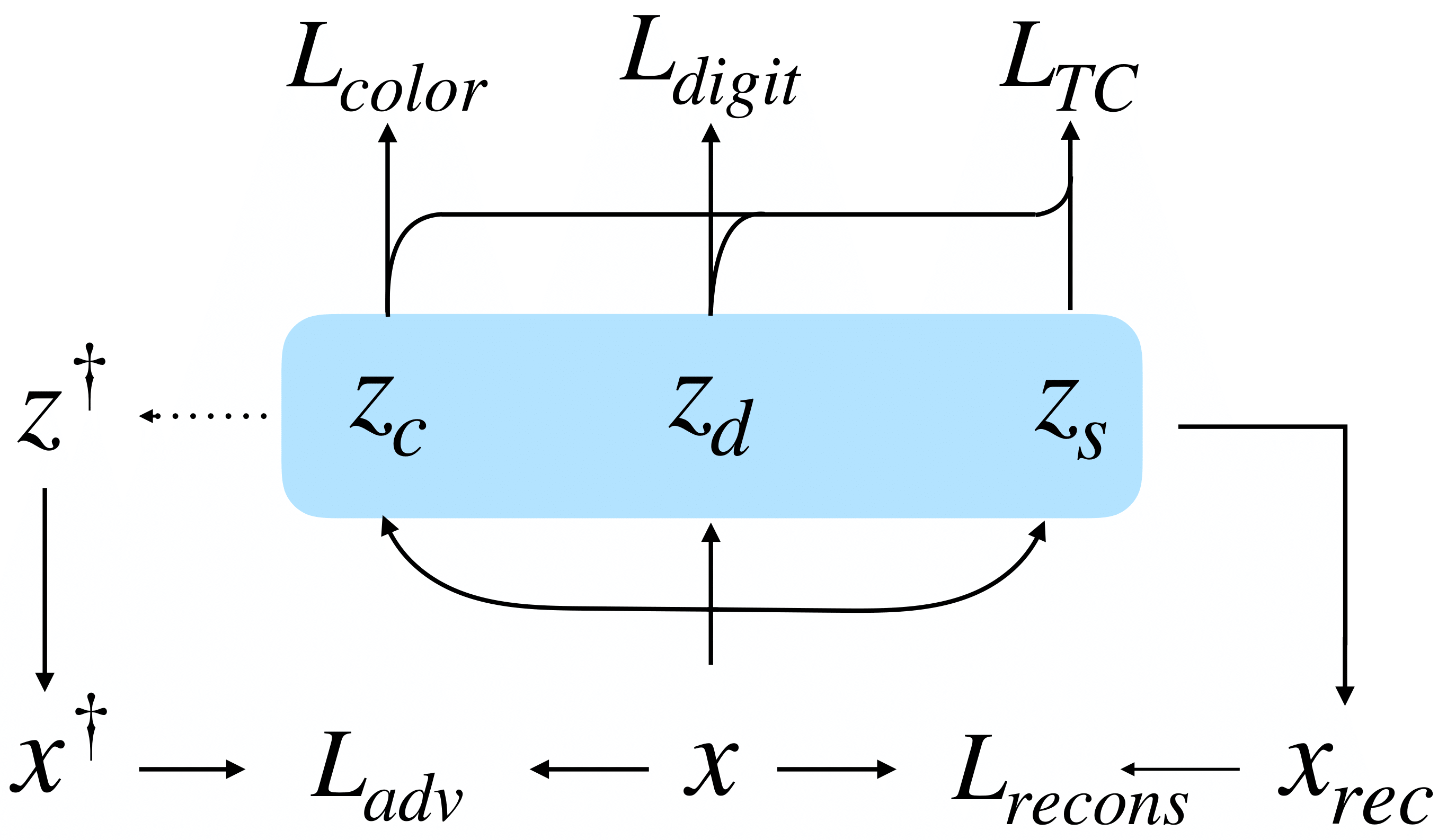}
\vspace{-1mm}
\caption{TC disentangling framework:  $\vx$ is an input image. $\vz_c$, $\vz_d$, $\vz_s$ are the \textit{color}, \textit{digit}, \textit{style} embeddings respectively. $\vx_\text{rec}$ is the reconstruction. $\vx^\dag$ is a generated sample from the shuffled latent space for adversarial training.}
\label{fig:disentangle_structure}
\vspace{-2mm}
\end{figure}
\begin{figure}[t]
\centering
\includegraphics[width=0.48\columnwidth]{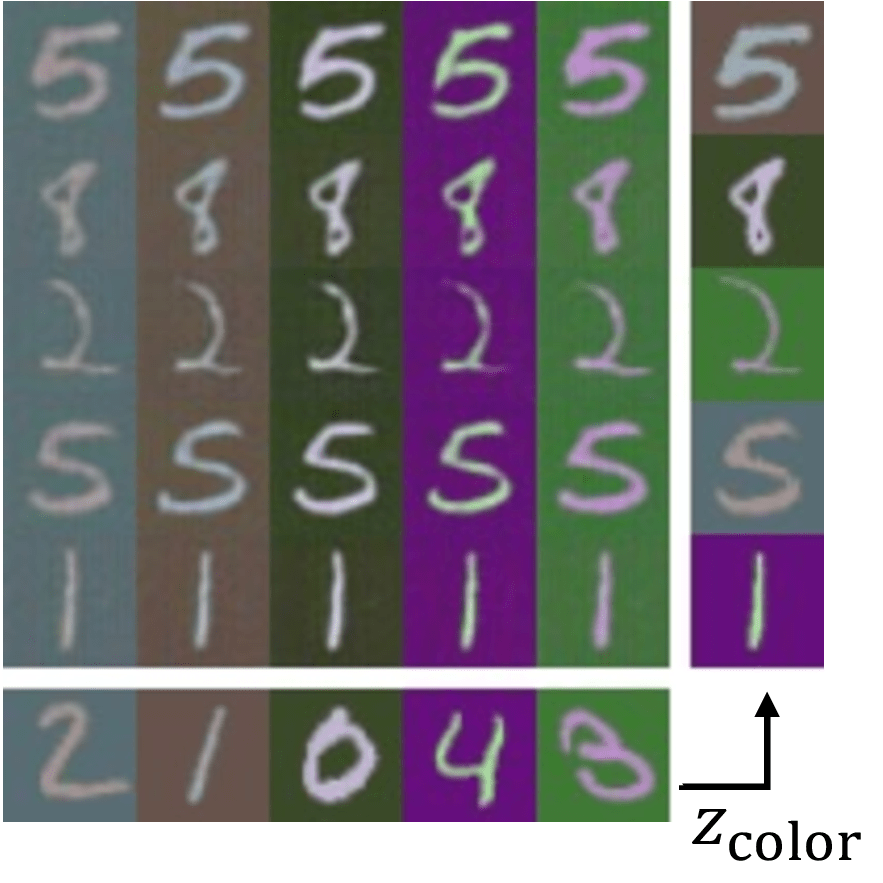}
\hspace{3pt}
\includegraphics[width=0.48\columnwidth]{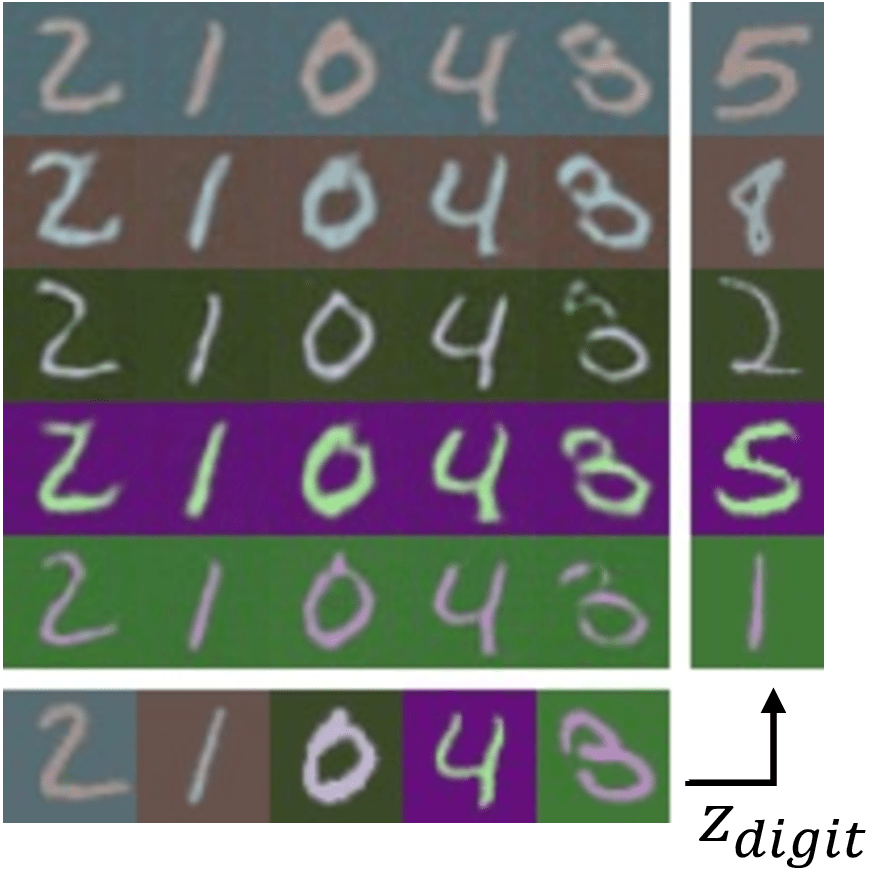}
\vspace{-2mm}
\captionof{figure}{Generated examples with latent embedding swapping: The latent embeddings \(z_{color}\)(left-side), \(z_{digit}\)(right-side) of the bottom row (source) are swapped to the corresponding embeddings of images in the right column (target).
}
\label{fig:disentangle_result}
\vspace{-4mm}
\end{figure}

For the TC minimization task, we conduct an experiment on a ColorMNIST hand-writing dataset~\citep{lecun1998gradient, esser2020disentangling} to learn disentangled digit, color, and style representations of each hand-written number image. Each input data $(\vx^i, \vy_d^i, \vy_c^i)$ contains three components: a image \(\vx^i\), its digit label \(\vy_d^i \in \{0,1,\ldots, 9\}\), and its color label vector \(\vy_c^i = (R^i, G^i, B^i)\), where $R^i, G^i, B^i\in [0,1]$ represents the intensity of colors (\textit{Red}, \textit{Green}, \textit{Blue}).
%

To learn the digit, color, and style representations from input images,  we use the neural encoder \(\calE(\cdot)\) to map each image $\vx_i$ to the corresponding latent representations \(\vz_{d}\), \(\vz_{c}\), \(\vz_{s}\), respectively. The digit representation $\vz^i_d$ is supposed to include sufficient digit information from $\vx^i$, hence we use a digit classifier $\calF_d(\cdot)$ to predict the digit label with loss $\calL_{digit} = \text{Cross-Entropy}(\calF_d(\vz^i_d), \vy_d^i)$. Similarly, we set a color regression function $\calF_c(\cdot)$ on the color embedding $\vz_{c}^i$ to ensure representativeness by minimizing $l$-2 norm $\calL_\text{color} = \Vert \calF_c(\vz_{c}^i) - \vy_{c}^i \Vert^2$. Excluding the digit and color information, the remaining part from the image should be the human hand-writing style information, which is assumed to be independent of the digit and color information. Therefore, we minimize the total correlation $\calL_\text{TC} = \TC(\vz_{d}, \vz_{c}, \vz_{s})$ as a regularizer to make sure different representation components do not include information from the others. Finally, we introduce a decoder $\calD(\cdot)$ to reconstruct the original image $\vx^i$ from $(\vz_d, \vz_c, \vz_s)$ to induce sufficient information into the latent representations with loss $\calL_\text{recons} = \Vert \calD(\vz_d, \vz_c, \vz_s) - \vx \Vert^2$.
To further enhance the generation quality of the decoder $\calD(\cdot)$ , we adapt an adversarial learning regularizer, 
 where  we randomize the combination of the latent representation components in each batch, 
 then treat the corresponding decoder output as artificial (synthesized) data \(\vx^\dag
\). 
With a discriminator \(\dis\), we use adversarial training \citep{goodfellow2014generative} to ensure the decoder to generate high-quality samples with loss \(\calL_\text{adv} = \E_\vx [\log \dis(\vx)] +  \E_{\vx^\dag} [\log (1 - \dis( \vx^\dag))] \). 
\Figref{fig:disentangle_structure} illustrates the whole framework.
More details about the framework are shown in the Supplementary Material.

\textbf{Evaluations }
To evaluate the quality of our disentangled representations, we conduct a controllable generation testing task to check whether the learned embedding components ($\vz_d$, $\vz_c$, or $\vz_s$) can control the corresponding attributes (digit, color, or style) of the generated sample $\calD(\vz_d, \vz_c, \vz_s)$. Hence, we consider three perspectives to evaluate the disentanglement: 
($\bm{i}$) \textit{digit transfer}:
Select another real sample $\vx'$ from testing set and obtain its digit embedding $\vz'_d$. Next, feed the new latent embedding combination $(\vz'_d, \vz_c, \vz_s)$ into the decoder to generate a sample $\vx^\dag = \calD(\vz'_{d}, \vz_{c}, \vz_{s})$, which is supposed to share the same digit information with $\vx'$, so we predict the digit label on $\vx^\dag$ then report the classification accuracy Acc$_d$.
($\bm{ii}$) \textit{digit preservation}: Select another testing sample $\vx''$, and replace $\vz_c$  with $\vz''_c$ to generate  $\vx^\ddag = \calD(\vz_{d}, \vz''_{c}, \vz_{s})$. Then predict the digit label of $\vx^\ddag$ and report the digit classification accuracy Acc$_c$.
($\bm{iii}$) \textit{color transfer}: For selected $\vx''$ in (${ii}$), we have its  color label vector $\vy''_c$. Hence, we can directly synthesize a $\bar{\vx}^\ddag$ by setting $\vy''_c$ on $\vx$.  We report the $l_2$-distance (Res.$l_2$) between generated $\vx^\ddag$ and synthetic $\bar{\vx}^\ddag$ as color transfer quality.

\begin{figure*}
\centering
  \begin{subfigure}{.18\textwidth}
  \centering
\includegraphics[width=.8\textwidth]{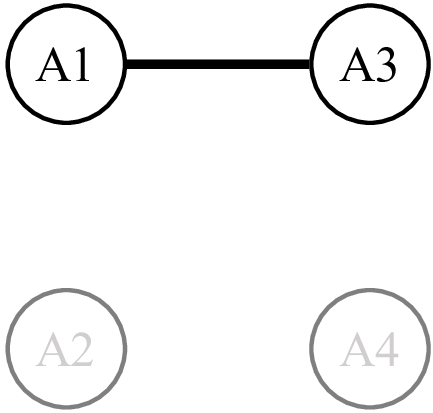}
\caption{MI-fix}
\end{subfigure}
  \begin{subfigure}{.18\textwidth}
   \centering
\includegraphics[width=.8\textwidth]{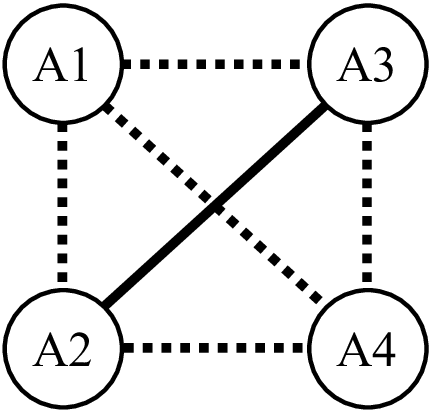}
\caption{MI-sample}
\end{subfigure}
  \begin{subfigure}{.18\textwidth}
   \centering
\includegraphics[width=0.8\textwidth]{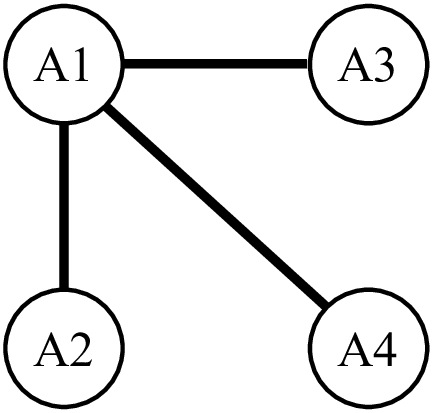}
\caption{Core-View}
\label{fig:core-view}
\end{subfigure}
  \begin{subfigure}{.18\textwidth}
   \centering
\includegraphics[width=0.8\textwidth]{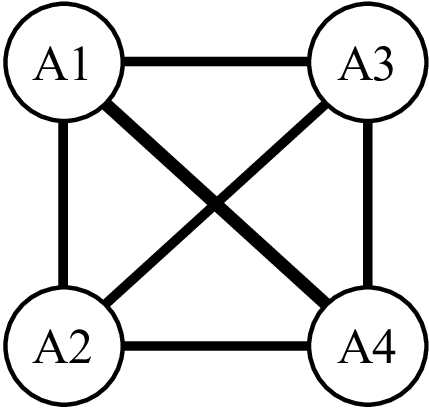}
\caption{Full-Graph}
\end{subfigure}
  \begin{subfigure}{.18\textwidth}
   \centering
\includegraphics[width=0.88\textwidth]{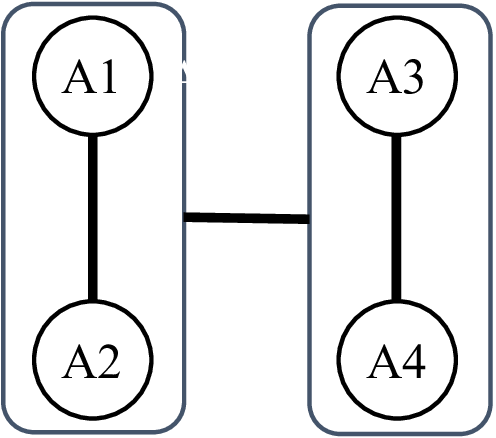}
  \caption{TC-Tree}
\end{subfigure}
\vspace{-2mm}
\caption{Correlation maximization strategy for multiple  augmentations. Nodes A$_1$, A$_2$, A$_3$, A$_4$ mean four input augmentations.
Each solid line indicates an MI maximization between the connected two augmentations.
Each dashed line means the connected  augmentation pair is randomly selected.
}
\label{fig:aug-structure}
\vspace{-2mm}
\end{figure*}

\begin{table}[t]
\caption{ Controllable generation results on ColorMNIST.
``$\checkmark$''  means adversary training is applied. 
Acc$_d$ and  Acc$_c$ are generated digit classification accuracy.  Res.$l_2$ measures the residual $l_2$-distance for color transferred samples.}
\vspace{-1mm}
\centering
\resizebox{0.79\columnwidth}{!}{
\begin{tabular}{@{}clrrr@{}}

\toprule
Adv.                 & Methods    & Acc$_d$ & Acc$_c$ & Res. $l_2$   \\ \midrule
\multirow{2}{*}{} & AE  &  2.70       & 91.65     &   168.95            \\
                    & TC (Ours)  &  92.70  &  97.69    & \textbf{28.66}        \\ \midrule

\checkmark& AE  & 0.14    & \textbf{98.82}   & 389.73        \\
\checkmark                          & 
VAE
& 20.24   & 92.21   & 254.83        \\
\checkmark                          & 
DIIN
&  94.57   & 97.45    & 79.38              \\
\checkmark                          & TC (Ours)  & \textbf{96.38}   & 98.23   & 43.04         \\ \bottomrule 
\end{tabular}
\label{table:disen_control_gen}
}
\vspace{-3mm}
\end{table}
\textbf{Results and Analysis }
We compared our method with vanilla auto-encoder (AE), variational auto-encoder (VAE)~\citep{kingma2013auto}, and DIIN~\citep{esser2020disentangling}, and report the aforementioned evaluation metrics in \Tableref{table:disen_control_gen}. Implementation and setup details are provided in the Supplementary Material.
Vanilla AE fails on Acc\(_d\) because without any disentangled regularizer,  \(\vz_s\) can contain the information revealed from \(\vz_d\) and  \(\vz_c\)
VAE partially solves the embedding entanglement problem with its KL divergence term in the learning objective, which encourages the latent embedding being close to a dimension-wise-independent standard Gaussian.   
DIIN\citep{esser2020disentangling} achievers more significant   improvement, by projecting the latent space of an auto-encoder to a multi-variate Gaussian distribution using normalizing flow \citep{kingma2018glow}. 
Our TC-based method uniformly outperforms the vanilla AE without the adversarial training. Among all method with adversarial training, the vanilla AE reaches the highest color classification accuracy, for which  our TC method is also strongly competitive. Moreover, our TC is in the lead on the other two metrics and leave a significant performance gap to the vanilla AE, which indicates our TC-based method can learn more balanced attribute embeddings in terms of representativeness and disentanglement. In addition, we show the generated image examples of digit transfer (\(i\)) and color transfer (\(iii\)) in Figure~\ref{fig:disentangle_result}, where 
both color and digit information can be successfully preserved in the transferred images.
\subsection{TC Maximization for Contrastive Representation Learning}
\vspace{-2mm}

\begin{table*}[t]
\centering
\caption{ Text representation evaluation on STS tasks (Spearman’s correlation with “all” setting). Methods with `` -TC'' means that total correlation is used as the loss function. 
}
\vspace{-2mm}
\resizebox{\textwidth}{!}{
\begin{tabular}{lcccccccc}
\toprule
Method                                          & STS12          & STS13          & STS14          & STS15          & STS16          & STS-B          & SICK-R         & Avg.                         \\ \midrule
IS-BERT$_{\textrm{base}}$~\citep{zhang2020unsupervised}                  & 56.77          & 69.24          & 61.21          & 75.23          & 70.16          & 69.21          & 64.25          & 66.58                        \\
ConSERT$_{\textrm{base}}$~\citep{yan2021consert}         & 64.64          & 78.49          & 69.07          & 79.72          & 75.95          & 73.97          & 67.31          & 72.74                        \\
SimCSE-BERT$_{\textrm{base}}$~\citep{gao2021simcse}     & 68.40          & 82.41          & 74.38          & 80.91          & 78.56          & 76.85          & 72.23          & 76.25  \\
SimCSE-BERT$_{\textrm{base}}$-TC     & 68.66 & 81.45 & 74.34 & 81.81 & 79.24 &    78.85     &      72.46      & 76.69
\\
PromptBERT$_{\textrm{base}}$ ~\citep{DBLP:journals/corr/abs-2201-04337}       & 71.56          & 84.58          & 76.98          & 84.47          & \textbf{80.60} & 81.60          & 69.87          & 78.54$_{\pm 0.15}$           \\
PromptBERT$_{\textrm{base}}$-TC                                & \textbf{72.05} & \textbf{84.61} & \textbf{77.23} & \textbf{84.73} & 80.34          & \textbf{81.89} & \textbf{70.23} & \textbf{78.72$_{\pm 0.10}$}  \\ \midrule
SimCSE-RoBERTa$_{\textrm{base}}$~\citep{gao2021simcse}    & 70.16          & 81.77          & 73.24          & 81.36          & 80.65          & 80.22          & 68.56          & 76.57                        \\
SimCSE-RoBERTa$_{\textrm{base}}$-TC    & 71.46 & 82.16 & 74.14 & 82.17 & 80.93 &    80.02     &      68.24      & 77.02 \\ 
PromptRoBERTa$_{\textrm{base}}$~\citep{DBLP:journals/corr/abs-2201-04337}    & \textbf{73.94} & 84.74          & 77.28          & 84.99          & 81.74 & 81.88          & 69.50          & 79.15$_{\pm 0.25} $          \\
PromptRoBERTa$_{\textrm{base}}$-TC                                 & 72.58 &	 \textbf{85.06} &	 \textbf{78.24} &	 \textbf{85.82} &	 \textbf{81.95} &	 \textbf{82.94} &	 \textbf{70.47} &	 \textbf{79.58$_{\pm 0.25}$} \\ 
\bottomrule
\end{tabular}
}
\label{table:sts-results}
\end{table*}
To test the performance of TCeld on TC maximization, we conduct a unsupervised text representation learning experiment following the SimCSE~\citep{gao2021simcse} setups. More specifically,  we aim to train a encoder $\calE(\cdot)$ to map each sentence $\vx$ into representative embedding $\vz=\calE(\vx)$. According to SimCSE (introduced in Section~\ref{sec:related-work}), one can learn in a unsupervised way the encoder $\calE(\cdot)$, by first generating several data augmentations $(\tilde{\vx}_1, \tilde{\vx}_2, \dots, \tilde{\vx}_n)$, then maximizing the correlation of corresponding $(\tilde{\vz}_1, \tilde{\vz}_2, \dots, \tilde{\vz}_n)$. Most of the previous contrastive learning methods~\citep{chen2020simple,gao2021simcse,DBLP:journals/corr/abs-2201-04337} focus on two-augmentation cases, where the mutual information between the two augmentations $\MI(\tilde{\vz}_1; \tilde{\vz}_2)$ is maximized for each input $\vx$. Moreover, none of current text contrastive learning method handles multi-augmentation ($n \geq 3$) situations. However, \citet{tian2020contrastive} point out that contrastive learning with  more augmentations can further enhance the latent embedding quality. Therefore, for testing our TC estimators while attempting the \textit{first} multi-augmentation text contrastive learning, we plan to maximize $\TC(\tilde{\vz}_1, \tilde{\vz}_2,\dots, \tilde{\vz}_n)$ to train the text encoder $\calE$.

\textbf{Model Frameworks }
We conduct our multi-augmentation text contrastive leaning by extending prior two-augmentation methods,  SimCSE \citep{gao2021simcse} and PromptBERT \citep{DBLP:journals/corr/abs-2201-04337} with \textit{four} augmentations for each text input. Both SimCSE and PromptBERT maximize the InfoNCE MI estimator between the two augmentation embeddings with a BERT\citep{devlin2018bert}-based pretrained text encoder. Correspondingly, our extended SimCSE-TC and PromptBERT-TC utilize the same encoder structure but maximize TC-InfoNCE of four augmentation embeddings for each sentence. Following SimCSE~\citep{gao2021simcse} and PromptBERT~\citep{ DBLP:journals/corr/abs-2201-04337}, we finetine the text encoder $\calE$ on pretrained  BERT$_{\textrm{base}}$~\citep{devlin2018bert} and RoBERTa$_{\textrm{base}}$~\citep{liu2019roberta}. Based on our observation in Section~\ref{sec:exp-tc-estimation}, tree-like decomposition empirically works better for InfoNCE estimator. Therefore, we select TC-Tree-InfoNCE for this text contrastive learning task.
More setup details can be found in the Supplementary Material.

\textbf{Evaluation}
Following previous work~\citep{gao2021simcse, DBLP:journals/corr/abs-2201-04337}, 
we evaluate models on 7 semantic textual similarity (STS) datasets: STS12~\citep{agirre2012semeval}, STS13~\citep{ agirre2013sem}, STS14~\citep{agirre2014semeval}, STS15~\citep{agirre2015semeval}, STS16~\citep{ agirre2016semeval}, STS Benchmark~\citep{cer2017semeval} and SICK-Relatedness~\citep{marelli2014sick}. 
The task is to predict the similarity (ranging from 0 to 5) between paired sentences with learned text representations.  We report Spearman's correlation~\citep{myers2013research} between model prediction and the ground truth. More details about the model design, hyperparameter settings and evaluation metrics are in the Supplementary Material.

\textbf{Results and Analysis }
We report the mean and standard deviation over $10$ runs with different random seeds in Table~\ref{table:sts-results}. On most of the evaluation datasets, our TC-based methods outperform their corresponding baselines, in which Prompt-BERT/RoBERTa are the  state-of-the-art unsupervised sentence representation learning methods. These results also underline the claim that more augmentations lead to higher representation quality in contrastive learning.  

\begin{table}[t]
\centering
\caption{Ablation study of different correlation maximization strategies. Average Spearman coefficient is reported over the seven STS tasks. }
\vspace{-2mm}
\centering
\resizebox{\columnwidth}{!}{
\begin{tabular}{@{}lll@{}}
\toprule
\diagbox{Method}{Model} & \multicolumn{1}{c}{\begin{tabular}[c]{@{}c@{}}Prompt\\ BERT\end{tabular}}             & \multicolumn{1}{c}{\begin{tabular}[c]{@{}c@{}}Prompt\\ RoBERTa\end{tabular}}          \\ \midrule
MI-fix                                & 78.54$_{\pm 0.15}$   & 79.28$_{\pm 0.23}$ \\
MI-sample                               & 78.41$_{\pm 0.13}$  & 79.27$_{\pm 0.34}$ \\
Core-View~\citep{tian2020contrastive}                                  & 78.64$_{\pm 0.22}$   & 79.41$_{\pm 0.25}$ \\
Full-Graph~\citep{tian2020contrastive}                                & 78.52$_{\pm 0.13}$  &  79.49$_{\pm 0.26}$ \\
TC                 & 78.72$_{\pm 0.10}$  & 79.58$_{\pm 0.25}$ \\ \bottomrule
\end{tabular}
}
 \label{table:sts-ablation-study}
\end{table}

For the ablation study, we fix the number of augmentations to $4$, and test the influence of different embedding correlation maximization strategies. Since there is no prior work on multi-augmentation text contrastive learning, we proposed several substitute strategies by ourselves in Figure~\ref{fig:aug-structure}:
(a) 
\textit{MI-fix}: only select the first two augmentations and omit the others; 
(b) \textit{MI-sample}: randomly select two augmentations and omit the others; 
(c) \textit{Core-View}~\citep{tian2020contrastive}:  calculate the MI between one fixed augmentation and each augmentation in the rest.
(d) \textit{Full-Graph}~\citep{tian2020contrastive}: calculate the MI values between each augmentation pairs. 
%
From the results in Table~\ref{table:sts-ablation-study}, MI-fix and MI-sample which only utilized two augmentations, have lower Spearman score than the other methods using all augmentations' information. Our TC-based correlation maximization strategy slightly outperforms Core-View and Full-Graph. Core-View does not consider the correlation among the other augmentations (A$_2$, A$_3$, A$_4$ in Figure~\ref{fig:core-view}). Full-Graph uses expensive MI estimators, which quadratically increases the computational complexity of augmentation correlation maximization, while being more prone to overfitting than TC with the same data size.






\vspace{-2mm}
\section{CONCUCLUSION}
\vspace{-2mm}
We have derived a Total Correlation Estimation with Linear Decomposition (TCeld), which converts total correlation into summation of mutual information terms along two separation paths (\textit{i.e.}, line-like and tree-like).
By applying variational estimators to each MI term in TCeld, we have obtained TC-Line and TC-Tree estimators. Further, we have analyzed the statistical properties of the proposed TC estimators and claimed their strong consistency when induced by appropriate MI estimators such as MINE, NWJ, and InfoNCE. Moreover, we have empirically demonstrated the effectiveness of the proposed TC estimators on both TC estimation and optimization tasks. The experimental results show that our TC estimators can only provide reliable estimation from samples, but also serve as an effective learning regularizer for model training. More properties of TCeld, such as sample complexity, unbiasedness, and low-variance, remain to be explored both theoretically and empirically.
%
We hope this study can serve to promote TC, as a multi-variate information concept, to apply into cutting-edge deep learning models, such as representation learning, controllable generation,  model distillation and ensemble.

\bibliography{reference}
\bibliographystyle{plainnat}
\clearpage

\appendix
\onecolumn 
\section{PROOFS}
\begin{proof}[Proof of Theorem~\ref{thm:general_connection}]
Note that $\mX_{\calA} := (\vx_{i_1}, \vx_{i_2}, \dots, \vx_{i_m} )$ and $\mX_{\hat{\calA}} = \mX / \mX_{\calA}$. Denote $\mX_{\hat{\calA}}= (\vx_{j_1}, \vx_{j_2}, \dots, \vx_{j_l})$. Then
\begin{align*}
&\mathcal{TC}(\mX) =\bbE_{p(\mX)} \left[\log \frac{p(\vx_1, \vx_2, \dots, \vx_n)}{p(
 \vx_1) p(\vx_2)\dots p(\vx_n) }\right] \\
 = & \bbE_{p(\mX)} \left[\log \left( \frac{p(\mX_\calA)}{p(\vx_{i_1})p(\vx_{i_2})\dots p(\vx_{i_{m}})} \cdot \frac{p(\mX_{\hat{\calA}})}{ p(\vx_{j_1})p(\vx_{j_2})\dots p(\vx_{j_{l}})} \cdot \frac{p(\mX)}{p(\mX_{\calA}) p(\mX_{\hat{\calA}})}
 \right)\right] \\
 = &\TC(\mX_{\calA}) + \TC(\mX_{\hat{\calA}}) + \MI(\mX_{\calA};\mX_{\hat{\calA}}) 
\end{align*}
\end{proof}

\begin{proof}[Proof of Corollary~\ref{thm:line-like}]
We denote $\mX_{i:j} := (\vx_i, \vx_{i+1}, \dots, \vx_{j-1}, \vx_{j} )$. Note that
\begin{align*}
 \mathcal{TC}(\mX_{1:n}) =& \bbE_{p(\vx_1, \vx_2, \dots, \vx_n)} \left[\log \frac{p(\vx_1, \vx_2, \dots, \vx_n)}{p(\vx_1) p(\vx_2)\dots p(\vx_n) }\right] \\
 = & \bbE_{p(\vx_1, \vx_2, \dots, \vx_n)} \left[\log \left( \frac{p(\vx_1, \vx_2, \dots, \vx_{n-1}, \vx_n)}{p(\vx_1,\vx_2,\dots,\vx_{n-1}) p(\vx_{n})} \cdot \frac{p(\vx_1, \vx_2, \dots, \vx_{n-1})}{ p(\vx_1)p(\vx_2)\dots p(\vx_{n-1})} \right)\right] \\
 = &\MI(\vx_1,\vx_2,\dots, \vx_{n-1} ; \vx_n) + TC(\mX_{1:n-1}) \\
 = &\MI(\mX_{1:n-1} ; \vx_n) + \TC(\mX_{1:n-1})
\end{align*}
Similarly, 
\begin{equation}
    \TC(\mX_{1:n}) =\MI(\mX_{1:n-1} ; \vx_n) + \MI(\mX_{1:n-2}; \vx_{n-1}) + \TC(\mX_{1:n-2}) = \sum_{i=1}^{n-1} \MI(\mX_{1:i}; \vx_{i+1})
\end{equation}
\end{proof}

\begin{proof}[Proof of Theorem~\ref{thm:linear-prop-unbiased}]
First consider line-like TC estimator $\widehat{\TC}_\text{Line}[\hat{\MI}] (\mX) = \sum_{i=1}^{n-1} \hat{\MI}(\mX_{1:i}; \vx_{i+1})$.
If the MI estimator $\hat{\MI}$ is unbiased, by definition, $\bbE[\hat{\MI}(\mX_{1:i}; \vx_{i+1})] = \MI(\mX_{1:i}; \vx_{i+1})$. Taking expectation for the TC estimator, 
\begin{equation}
    \bbE [\widehat{\TC}_\text{Line}[\hat{\MI}] (\mX) ]= \sum_{i=1}^{n-1} \bbE[\hat{\MI}(\mX_{1:i}; \vx_{i+1})] = \sum_{i=1}^{n-1} \MI(\mX_{1:i}; \vx_{i+1}) = \TC(\mX), 
\end{equation}
which means that $\widehat{\TC}_\text{Line}[\hat{\MI}] (\mX)$ is unbiased. Similarly, we can show that $\widehat{\TC}_\text{Tree}[\hat{\MI}] (\mX)$ is unbiased. 
\end{proof}

\begin{proof}[Proof of Theorem~\ref{thm:linear-prop-consistent}]
For the convenience of notation, we show the proof with our $\widehat{\TC}_\text{Line}[\hat{\MI}]$ estimator. The proof can be easily applied to $\hat{\TC}_\text{Tree}[\hat{\MI}]$, since both $\hat{\TC}_\text{Line}[\hat{\MI}]$ and $\hat{\TC}_\text{Tree}[\hat{\MI}]$ are linear combination of MI terms. We denote $\hat{\MI}_{m} (\vx ;\vy)  = \hat{\MI} (\{(\vx^k, \vy^k)\}_{k=1}^m)$ and $\TCeld[\hat{\MI}]_m(\mX) = \TCeld[\hat{\MI}](\{\mX^k\}_{k=1}^m)$ as the estimators with $m$ samples $\{(\vx^k, \vy^k)\}_{k=1}^m  \sim p(\vx,\vy)$, and $\{\mX^k\}_{k=1}^m \sim p(\mX)$ respectively.

\textbf{Strong Consistency: } If $\hat{\MI}$ is a strong consistent estimator, by the Definition~\ref{def:strong-consistency}, $\forall \varepsilon > 0$, with a fix variable dimension $n \in \bbN_{+}$, for each variable pair $(\mX_{1:i}, \vx_{i+1})$ ($i=1,2,\dots,n-1)$,  $\exists M_i>0$, such that $\forall m > M_i$,
\begin{equation}\label{eq:appendix-mi-consistent}
\bbP\Big\{\left|\hat{\MI}_{m}(\mX_{1:i}; \vx_{i+1}) - \MI(\mX_{1:i} ; \vx_{i+1})\right| \leq \frac{\varepsilon}{n}\Big\} = 1.    
\end{equation}
Let $\bar{M} = \max \{M_1, M_2, \dots, M_{n-1} \}$, then $\forall m > \bar{M}$, 
\begin{align}
&\bbP\Big\{\left|\TCeld_\text{Line}[\hat{\MI}]_m (\mX) - \TC(\mX)\right| \leq \varepsilon \Big\} \label{eq:prob-TC-consistent-form} \\
= & \bbP \Big\{ \left| \sum_{i=1}^{n-1} \hat{\MI}_m(\mX_{1:i}; \vx_{i+1})- \sum_{i=1}^{n-1}  \MI(\mX_{1:i}; \vx_{i+1})\right| \leq \varepsilon \Big\} \label{eq:prob-TC-leq-eps} \\  
\geq& \bbP\left[ \bigcap_{i=1}^{n-1} \Big\{\left|\hat{\MI}_{m}(\mX_{1:i}; \vx_{i+1}) - \MI(\mX_{1:i} ; \vx_{i+1})\right| \leq \frac{\varepsilon}{n}\Big\}\right]. \label{eq:prob-each-mi-leq-eps}
\end{align}
The inequality between \eqref{eq:prob-TC-leq-eps} and \eqref{eq:prob-each-mi-leq-eps} is because the condition in \eqref{eq:prob-TC-leq-eps} is sufficient to deduce the condition in \eqref{eq:prob-TC-leq-eps}:
\begin{align*}
    &\left| \sum_{i=1}^{n-1} \hat{\MI}_m(\mX_{1:i}; \vx_{i+1})- \sum_{i=1}^{n-1}  \MI(\mX_{1:i}; \vx_{i+1})\right|  \\ \leq & \sum_{i=1}^{n-1}\left|\hat{\MI}_{m}(\mX_{1:i}; \vx_{i+1}) - \MI(\mX_{1:i} ; \vx_{i+1})\right| \leq \frac{n-1}{n} \varepsilon < \varepsilon.
\end{align*}
Denote event $\calB_i = \Big\{\left|\hat{\MI}_{m}(\mX_{1:i}; \vx_{i+1}) - \MI(\mX_{1:i} ; \vx_{i+1})\right| \leq \frac{\varepsilon}{n}\Big\}$, by \eqref{eq:appendix-mi-consistent}, $\bbP[\calB_i] = 1$.
Consider the union $\calB_i \cup \calB_j$, we have:
\begin{equation}
    1 = \bbP[\calB_i] \leq \bbP[\calB_i \cup \calB_j] = \bbP[\calB_i] + \bbP[\calB_j] - \bbP[\calB_i \cap \calB_j] = 2 - \bbP[\calB_i \cap \calB_j] \leq 1,
\end{equation}
which means $\bbP[\calB_i \cap \calB_j] = 1$. Iteratively applying this conclusion, we know $\bbP[\cap_{i=1}^{n-1} \calB_i] = 1$. Therefore, $\bbP\left[ \bigcap_{i=1}^{n-1} \Big\{\left|\hat{\MI}_{m}(\mX_{1:i}; \vx_{i+1}) - \MI(\mX_{1:i} ; \vx_{i+1})\right| \leq \frac{\varepsilon}{n}\Big\}\right] = \bbP[\cap_{i=1}^{n-1} \calB_i] = 1$. Combining with \eqref{eq:prob-TC-consistent-form} and \eqref{eq:prob-each-mi-leq-eps}, we conclude that $\forall \varepsilon$, $\exists \hat{M}$, such that $\forall m > \hat{M}$, $\left|\TCeld_\text{Line}[\hat{\MI}]_m (\mX) - \TC(\mX)\right| \leq \varepsilon$ almost surely, which supports $\TCeld_\text{Line}[\hat{\MI}]$ is strongly consistent.

\textbf{Consistency: } If $\hat{\MI}$ is a consistent estimator, by the Definition~\ref{def:consistency}, $\forall \varepsilon>0$ and $\sigma >0$, with a fixed variable dimension $n\in \bbN_{+}$, for each variable pair $(\mX_{1:i}, \vx_{i+1})$ ($i=1,2,\dots,n-1$), $\exists M_i $, such that $\forall m > M_i$, 
\begin{equation}
\bbP \Big\{\left|\hat{\MI}_m(\mX_{1:i}; \vx_{i+1}) - \MI(\mX_{1:i}; \vx_{i+1}) \right| < \frac{\varepsilon}{n} \Big\} > 1- \frac{\sigma}{n}.
\end{equation}

Let $\hat{M} = \max\{M_1, M_2, \dots, M_{n-1}\}$, then $\forall m > \hat{M}$, similar to \eqref{eq:prob-TC-consistent-form}, \eqref{eq:prob-TC-leq-eps}, and \eqref{eq:prob-each-mi-leq-eps},
\begin{align}
&\bbP\Big\{\left|\TCeld_\text{Line}[\hat{\MI}]_m (\mX) - \TC(\mX)\right| < \varepsilon \Big\} \label{eq:prob-TC-w-consistent-form} \\
= & \bbP \Big\{ \left| \sum_{i=1}^{n-1} \hat{\MI}_m(\mX_{1:i}; \vx_{i+1})- \sum_{i=1}^{n-1}  \MI(\mX_{1:i}; \vx_{i+1})\right| <\varepsilon \Big\} \label{eq:prob-TC-w-leq-eps} \\  
\geq& \bbP\left[ \bigcap_{i=1}^{n-1} \Big\{\left|\hat{\MI}_{m}(\mX_{1:i}; \vx_{i+1}) - \MI(\mX_{1:i} ; \vx_{i+1})\right| < \frac{\varepsilon}{n}\Big\}\right]. \label{eq:prob-w-each-mi-leq-eps}
\end{align}
Denote $\calB_i = \Big\{\left|\hat{\MI}_{m}(\mX_{1:i}; \vx_{i+1}) - \MI(\mX_{1:i} ; \vx_{i+1})\right| < \frac{\varepsilon}{n}\Big\}$. Since $1 \geq \bbP[\calB_i \cup \calB_j] = \bbP[\calB_i] + \bbP[\calB_j] - \bbP[\calB_i \cap \calB_j]$, we have $\bbP[\calB_i \cap \calB_j] \geq \bbP[\calB_i] + \bbP[\calB_j] - 1  > (1 -\frac{\sigma}{n}) + (1- \frac{\sigma}{n}) - 1 = 1 - \frac{2}{n}\sigma$. Similarly, we can obtain $\bbP[\calB_i \cap \calB_j \cap \calB_k] \geq \bbP[\calB_i \cap \calB_j] + \bbP[\calB_k] -1 \geq (1-\frac{2}{n}\sigma ) +(1-\frac{\sigma}{n}) = 1 - \frac{3}{n}\sigma$ and $\bbP\{\left|\TCeld_\text{Line}[\hat{\MI}]_m (\mX) - \TC(\mX)\right| < \varepsilon \} \geq \bbP\left[\cap_{i=1}^{n-1}\calB_i\right] \geq 1 - \sigma$.
Therefore, $\forall \varepsilon > 0$, $\lim_{m\rightarrow \infty} \bbP\{|\TCeld_\text{Line}[\hat{\MI}]_m(\mX) - \TC(\mX)| \geq \varepsilon\} = 0$.  $\TCeld_\text{Line}[\hat{\MI}]$ is consistent.
\end{proof}

\begin{proof}[Proof of Corollary~\ref{thm:TC-infonce-and-nwj-consistent}]
By Theorem~\ref{thm:linear-prop-consistent}, only need to show both $\hat{\MI}_\text{InfoNCE}$ and $\hat{\MI}_\text{NWJ}$ are strongly consistent. Inspired by the proof of Theorem~2 in \citet{belghazi2018mutual}, we only need to proof the following two lemmas:

\begin{lemma}\label{thm:lemma1}
For any $\eta > 0$, there exists a feedforward score network function $\hat{\phi}: \Omega \rightarrow \bbR$ such that $|\MI(\vx,\vy)- \hat{\MI}[\hat{\phi}] | \leq \eta$, where $\hat{\MI} \in \{ \hat{\MI}_\text{InfoNCE}, \hat{\MI}_\text{NWJ}\}$, $\hat{\MI}_\text{InfoNCE}[\hat{\phi}] =  \bbE_{p(\vx,\vy)}[\hat{\phi}(\vx,\vy)] - \bbE_{p(\vx)} [\log \bbE_{p(\vy)} [\exp\hat{\phi}(\vx, \vy)] ]$ and $\hat{\MI}_\text{NWJ}[\hat{\phi}] = \bbE_{p(\vx,\vy)} [\hat{\phi} (\vx,\vy) ] - \bbE_{p(\vx)p(\vy)} [\exp(\hat{\phi}(\vx,\vy) - 1 )]$
\end{lemma}

\begin{lemma}\label{thm:lemma2}
For any $\eta > 0$, let $\mathcal{H}$ be the family of functions $\phi: \Omega \rightarrow \bbR$ defined by a give network architecture. Assume the parameter $\theta$ of network $\phi$ are restricted to some compact domain $\Theta \subset \bbR^k$. Then there exists $N \in \bbN_{+}$, such that, $\forall m \geq N$, $|\hat{\MI}_m(\vx,\vy) - \sup_{\phi \in  \calH} \hat{\MI}[\phi]| \leq \eta$ with probability one. Here $\hat{\MI} \in \{ \hat{\MI}_\text{InfoNCE}, \hat{\MI}_\text{NWJ}\}$.
\end{lemma}

To proof Lemma~\ref{thm:lemma1}, for NWJ, we select function $\phi^*(\vx, \vy) = 1 + \log \frac{p(\vx,\vy)}{p(\vx)p(\vy)}$. Then $\bbE_{p(\vx,\vy)}[\hat{\MI}_\text{InfoNCE}[\phi^*] ] = \MI(\vx,\vy) $. The difference can be calculated as 
\begin{equation} \label{eq:appendix-decompose-mi-NWJ}
\MI(\vx, \vy) - \hat{\MI}_\text{NWJ}[\phi] = \bbE_{p(\vx, \vy)} [\phi^* - \phi] - \exp(-1) \bbE_{p(\vx)p(\vy)} [\exp(\phi^*) - \exp(\phi)].
\end{equation}
The right-hand side of \eqref{eq:appendix-decompose-mi-NWJ} has the same form as equation~(25) in \citet{belghazi2018mutual}, except a coefficient $\exp(-1)$ for the second term. Therefore, we can exactly follow the proof of Section~6.2.1 of \citet{belghazi2018mutual} to prove our Lemma~\ref{thm:lemma1} for the NWJ estimator, with only adjusting the coefficient weight of term $|\phi^* -\phi|$ and term $|\exp(\phi^*) - \exp(\phi)|$.

For InfoNCE, we select $\phi^*(\vx, \vy) =\log p(\vx|\vy) $, so that $\bbE_{p(\vx,\vy)}[\hat{\MI}_\text{InfoNCE}[\phi^*]] = \MI(\vx; \vy)$.
The difference can be written as
\begin{align}
    \MI(\vx; \vy) - \hat{\MI}_\text{InfoNCE}[\phi] =& \bbE_{p(\vx,\vy)} [\phi^* - \phi]  - \bbE_{p(\vx)} [\log \bbE_{p(\vy)} [\exp(\phi^*)] -\log \bbE_{p(\vy)} [\exp(\phi)]].
\end{align}
Similarly to Section~6.2.1 in \citep{belghazi2018mutual}, we can consider the cases whether $\phi$ is bounded, then apply the universal approximation theorem to show Lemma~\ref{thm:lemma1} for InfoNCE.

To proof Lemma~\ref{thm:lemma2}, we denote $\bbP = p(\vx,\vy) , \bbQ = p(\vx)p(\vy)$, and $\bbP_m, \bbQ_m$ for emprical distribution with $m$ samples. For NWJ, we calculate the difference
\begin{equation}
    |\hat{\MI}_m(\vx,\vy) - \sup_{\phi \in  \calH} \hat{\MI}[\phi]| \leq \sup_{\phi \in \calH} |\bbE_{\bbP} [\phi] - \bbE_{\bbP_m} [\phi] | + \exp(-1)  \sup_{\phi \in \calH} |\bbE_{\bbQ} [\exp(\phi)] - \bbE_{\bbQ_m}[\exp(\phi)]|,
\end{equation} where the second term of right-hand side has the same form as equation~(32) in Secion~6.2.2 in \citet{belghazi2018mutual}. Therefore, we can follow the proof in Secion~6.2.2 of \citet{belghazi2018mutual} to prove Lemma~2 for NWJ. Similarly, the similar proving process can be applied to InfoNCE.
\end{proof}

\section{EXPERIMENTAL DETAILS}
All experiments are executed on a single NVIDIA Titan Xp GPU with 12,196M memory.


\subsection{TC estimation}
\paragraph{Experiment Design}
Mutual information between two multivariate Guassian distributions $X_1 \sim \mathcal{N}(0,\Sigma_1)$, $X_2 \sim \mathcal{N}(0,\Sigma_2)$ is 
$\frac{1}{2}\log \frac{\det(\Sigma_1)\det(\Sigma_2)}{\det(\Sigma)}$,
where $\Sigma$ is the covariance matrix of the joint distribution \([X_1,X_2]\). 

In our setting, our training samples are sampled from a joint distribution with \(n\) variables \([\vX_{i}, i\in \{0,1,2,3,n-1\}]\), the marginal distribution of each variable have zero mean and identity covariance matrix with dimension \(d\). Therefore, the determinant of covariance matrix of single variable \(\Sigma_i\) is 1 and the mutual information between two variables \(i,j\) is \(-\frac{1}{2}\log \det(\Sigma_{ij})\). The total correlation among variables are \(-\frac{1}{2}\log \det(\Sigma)\). To prove this, we use the idea of line-like structure. Assume that the total correlation of first \(k\) variables are  \(-\frac{1}{2}\log \det(\Sigma_{[:k]})\), the total correlation of the first \(k+1\) variables are 
\begin{align}
    -\frac{1}{2}\log \det(\Sigma_{[:{k+1}]}) = -\frac{1}{2}\log \det(\Sigma_{[:k]}) + \frac{1}{2}\log \frac{\det(\Sigma_{[:k]})}{\det(\Sigma_{[:{k+1}]})},
\end{align}
where \(\Sigma_{:k}\) is the covariance matrix of the first \(k\) variables. 


In our proof-of-concept experiments, we set \(n = 4, d = 10\).
The covariance matrix is
\[
\begin{bmatrix}
\sI_{d} & \sigma \sI_{d} & 0 & 0 \\
\sigma \sI_{d} & \sI_{d} & 0 & 0 \\
0 & 0 & \sI_{d} & \sigma \sI_{d} \\
0 & 0 &  \sigma \sI_{d} & \sI_{d}  
\end{bmatrix}
\]
The total correlation under such a design is \(-d\log(1-\sigma^2)\).

As we mentioned in the paper,  we can adjust the correlation coefficients (non-diagonal elements) \(\sigma\) to set the ground truth TC values in the set \(\{2.0, 4.0, 6.0, 8.0, 10.0\}\). 


\paragraph{Hyper-parameters}
All MI lower bounds require the learning of a value function $f(\vx,\vy)$; the CLUB upper bound requires the learning of a network approximation $q_\theta(\vy| \vx)$. To make a fair comparison, we set the value function and the neural approximation with one hidden layer and the same hidden units. For the multivariate Gaussian setup, the number of hidden units is $20$.
On the top of the hidden layer output, we add the ReLU activation function.
The learning rate for all estimators is set to $1\times 10^{-4}$. 
\paragraph{Bias, Variance, Mean Squared Error}
Figure~\ref{fig:line-like-mse}~\ref{fig:tree-like-mse} show the bias, variance and mean squared error using different mutual information estimators. The explanation to this figure is shown in the main paper.
\begin{figure}[t]
\begin{minipage}{\textwidth}
    \includegraphics[width=0.9\textwidth]{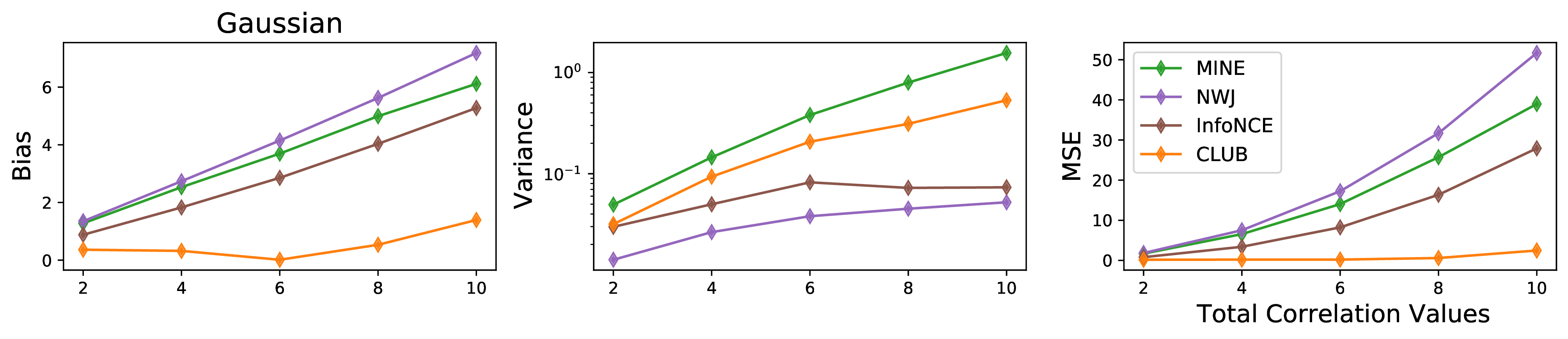}
    \vspace{-2mm}
    \caption{Bias, variance and MSE of \textbf{line-like} TC estimators.}
    \label{fig:line-like-mse}
\end{minipage}
\begin{minipage}{\textwidth}
    \includegraphics[width=0.9\textwidth]{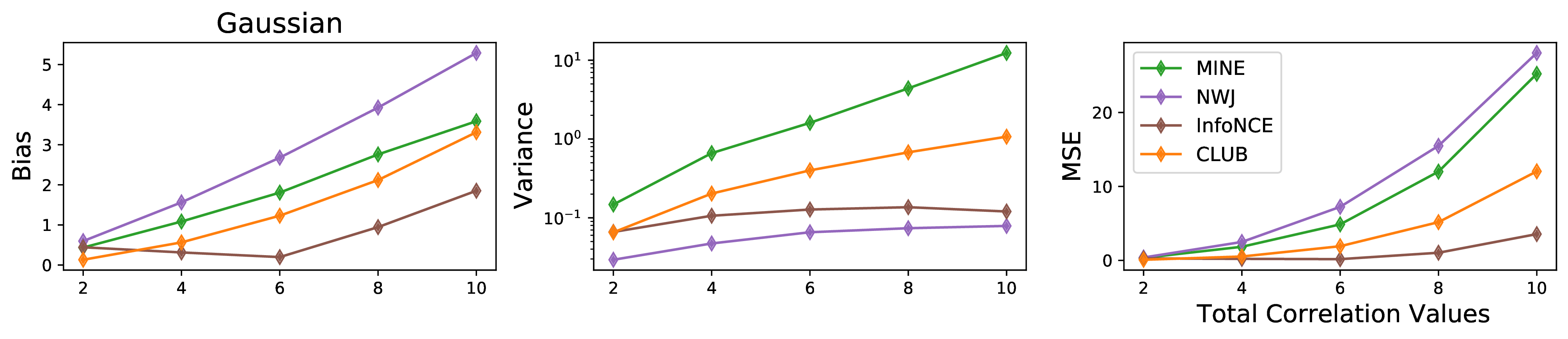}
    \vspace{-4mm}
    \caption{Bias, variance and MSE of \textbf{tree-like} TC estimators.}
    \label{fig:tree-like-mse}
\end{minipage}
\end{figure}

\paragraph{Non-Parametric Methods}

Since the probability can be directly estimated using non-parametric methods, we can also use non-parametric methods to estimate the total correlation directly. Here we compare with KDE, \textit{k}-NN based methods \citep{kraskov2004estimating, pal2010estimation} and its variant~\citep{gao2018demystifying}. As shown table~\Tableref{table:app:non}, the non-parametric methods can only estimate the total correlation decently when the dimension of input $n$ is low. Note that the experiment setting is exactly the same as the main paper except for dimension. 

\begin{table}[t]
\caption{Comparison between ours(Tree-based CLUB) with other non-parametric methods under two circumstance with different data dimensionality ($n = 2$ and $n = 10$. The table shows the absolute error between predicted and ground truth. \textbf{Bold} means the minimal error and the best predication.}
\begin{tabular}{@{}lllllllllll@{}}
\toprule
                                                  & \multicolumn{5}{c}{$n = 2$}                                                                                                        & \multicolumn{5}{c}{$n = 10$}                                                                                           \\ \midrule
\multicolumn{1}{l|}{Total Correlation}            & \multicolumn{1}{c}{2} & \multicolumn{1}{c}{4} & \multicolumn{1}{c}{6} & \multicolumn{1}{c}{8} & \multicolumn{1}{c|}{10}            & \multicolumn{1}{c}{2} & \multicolumn{1}{c}{4} & \multicolumn{1}{c}{6} & \multicolumn{1}{c}{8} & \multicolumn{1}{c}{10} \\ \midrule
\multicolumn{1}{l|}{KNN (\citep{kraskov2004estimating, pal2010estimation}} & 0.08                  & 0.29                  & 0.66                  & 1.29                  & \multicolumn{1}{l|}{2.27}          & 1.34                  & 2.66                  & 4.05                  & 5.47                  & 6.98                   \\
\multicolumn{1}{l|}{Bias-improved KNN (\citep{gao2018demystifying})} & 0.2                   & 0.59                  & 1.24                  & 2.21                  & \multicolumn{1}{l|}{3.54}          & 1.48                  & 2.97                  & 4.51                  & 6.07                  & 7.67                   \\
\multicolumn{1}{l|}{Kernel Density Estimation}    & 1.44                  & 1.40                  & 1.40                  & 1.41                  & \multicolumn{1}{l|}{\textbf{1.37}} & 9.02                  & 8.89                  & 8.98                  & 9.00                  & 8.97                   \\
\multicolumn{1}{l|}{Ours}                         & \textbf{0.03}         & \textbf{0.04}         & \textbf{0.55}         & \textbf{0.88}         & \multicolumn{1}{l|}{1.63}          & \textbf{0.23}         & \textbf{0.37}         & \textbf{0.22}         & \textbf{1.03}         & \textbf{1.96}          \\ \bottomrule
\end{tabular}
\label{table:app:non}
\end{table}

Considering MINE~\citep{belghazi2018mutual} has already performs better than \textit{k}-NN based methods~\citep{kraskov2004estimating, pal2010estimation} and MINE is one of our baselines, our performance should perform better. This is also reflected in the table. Meanwhile, \textit{k}-Nearest Neighbor estimation method focuses on total correlation estimation only with non-differentiable operation. The estimated value is devoted to variable independence and correlation analysis. While in our neural network-based methods, we calculate the total correlation in a derivative way and get meaningful gradient information for efficient back-propagation. The results are shown in our disentangle and representation learning experiments.

\subsection{TC Maximization: Disentanglement}
\label{sup:sec:tcmax}
\paragraph{Model Design}


\begin{table}[t]
\caption{Structure of the encoder \(\calE(\cdot)\), decoder \(\calD(\cdot)\) described in \Secref{sup:sec:tcmax}. The terms in brackets of Conv2d and DeConv2d are (input channel, output channel, filter size, stride size, zero padding size, bias included). IN means instance normalization. DeConv2d represents the deconvolution operator. }
\centering
\begin{tabular}{@{}lll@{}}
\toprule
  & Encoder                                         & Decoder                                           \\ \midrule
0 & Conv2d(3, 64, 4, 2, 1, False), IN, LeakyReLU    & DeConv2d(128, 512, 2, 1, 0, True)                 \\
1 & Conv2d(64, 128, 4, 2, 1, False), IN, LeakyReLU  & DeConv2d(512, 256, 4, 2, 1, False), IN, LeakyReLU \\
2 & Conv2d(128, 256, 4, 2, 1, False), IN, LeakyReLU & DeConv2d(256, 128, 4, 2, 1, False), IN, LeakyReLU \\
3 & Conv2d(256, 512, 4, 2, 1, False), IN, LeakyReLU & DeConv2d(128, 64, 4, 2, 1, False), IN, LeakyReLU  \\
4 & Conv2d(512, 128, 2, 1, 0, True), ,              & DeConv2d(256, 128, 4, 2, 1, False), ,Tanh         \\ \bottomrule
\end{tabular}
\label{sup:table:encoder-decoder}
\end{table}
\begin{table}[t]
\caption{Structure of classifier \(\calF_d(\cdot)\) and regressor \(\calF_c(\cdot)\) described in \Secref{sup:sec:tcmax}. The terms in the bracket of Linear are (input dimension, output dimension, bias included). }
\centering
\begin{tabular}{@{}lll@{}}
\toprule
  & Color Regressor              & Digit Classifier           \\ \midrule
0 & Linear(32, 16, True), ReLU   & Linear(32, 16, True), ReLU \\
1 & Linear(16, 3, True), Sigmoid & Linear(16, 10, True),    
\\ \bottomrule
\end{tabular}
\label{sup:table:digit-color}
\end{table}

The dimension of latent space \(\vz , \vz_{color}, \vz_{digit}\) is \(128, 32, 32 \). The structures of each model are shown in ~\Tableref{sup:table:encoder-decoder}~\ref{sup:table:digit-color}~\ref{sup:table:adv-tc}. We use TC-InfoNCE as our total correlation estimator. Note that when the number of variables is three. There is no difference between tree-based and line-based methods. Both of which require two mutual information estimators, one for mutual information estimation of any pair, the other is to capture the mutual information between the left one and the preselected pair. The parameters of these two estimators are listed in ~\Tableref{sup:table:adv-tc}. 

We compare our method with autoencoder and variational autoencoder. Apart from the reconstruction loss \(\calL_{recons}\) and KL loss used in AE and VAE, we also take regression loss \(\calL_{color}\), classification loss \(\calL_{digit}\) into account. We illustrate the failure transfer case under the AE setup to show the difference without the total correlation term in Figure~\ref{fig:ae}. Both the color and digit are not successfully transferred from the source (bottom row) to the target (rightmost column).

\begin{figure}[t]
     \centering
     \begin{subfigure}[ht]{0.3\textwidth}
         \centering
         \includegraphics[width=\textwidth]{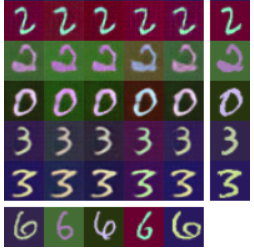}
         \caption{color transfer}
         \label{fig:y equals x}
     \end{subfigure}
    \hspace{10pt}
     \begin{subfigure}[ht]{0.3\textwidth}
         \centering
         \includegraphics[width=\textwidth]{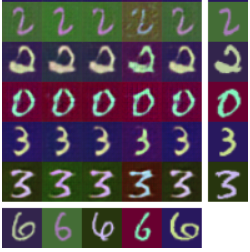}
         \caption{digit transfer}
         \label{fig:three sin x}
     \end{subfigure}
        \caption{Failure case of disentangle transfer, using autoencoder without the total correlation constraint. The bottom row is the source image and rightmost column is the target image.}
        \label{fig:ae}
\end{figure}

\begin{table}[t]
\caption{Structure of the adversarial discriminator \(\calH(\cdot)\) and two total correlation estimators described in \Secref{sup:sec:tcmax}. The terms in brackets of Conv2d and DeConv2d are (input channel, output channel, filter size, stride size, zero padding size, bias included). The terms in the bracket of Linear are (input dimension, output dimension, bias included). BN means batch normalization.}
\centering
\begin{tabular}{@{}llll@{}}
\toprule
  & Discriminator                            & TC estimator 1                & TC estimator 2                \\ \midrule
0 & Conv2d(3, 32, 5, 1, 2, False), BN, ReLU  & Linear(128, 90, True), ReLU   & Linear(64, 64, True), ReLU    \\
1 & Conv2d(32, 64, 5, 1, 2, False), BN, ReLU & Linear(90, 90, True), ReLU    & Linear(64, 64, True), ReLU    \\
2 & Linear(4096, 1), ,                      & Linear(90, 1, True), Softplus & Linear(64, 1, True), Softplus \\ \bottomrule
\end{tabular}
\label{sup:table:adv-tc}
\end{table}

\paragraph{Hyper-parameters}
The total training epoch and batch size are 300 and 512. 
Learning rate is $5\mathrm{e}{-3}$ for encoder and decoder, $1\mathrm{e}{-4}$ for the adversarial discriminator, $1\mathrm{e}{-3}$ for the digit classifier and color regressor. All the optimizers are Adam. Since the final loss is the summation of each term, the weight of each term is one except for the color regressor loss, which is ten. We also add spectral normalization~\citep{miyato2018spectral} to the adversarial discriminator to further stabilize the adversarial training.

\begin{table}[t!]
\caption{Four templates used in PromptBERT-$\TC$ and PromptRoBERTa-$\TC$. [CLS],[MASK], [SEP] are special tokens used in BERT. [X] is the placeholder of the input sentence.}
\begin{adjustbox}{width=\textwidth,center}
\begin{tabular}{@{}lll@{}}
\toprule  & BERT  & RoBERTa \\ \midrule
Template 0 & {[}CLS{]} This sentence of "{[}X{]}" means{[}MASK{]}.{[}SEP{]}                    & {[}CLS{]} This sentence of '{[}X{]}' means{[}MASK{]}.{[}SEP{]}                    \\
Template 1 & {[}CLS{]} The sentence : "{[}X{]}" means {[}MASK{]}.{[}SEP{]}                     & {[}CLS{]} The sentence : '{[}X{]}' means {[}MASK{]}.{[}SEP{]}                     \\
Template 2 & {[}CLS{]} The sentence ' {[}X{]} ' has the same meaning with {[}MASK{]}.{[}SEP{]} & {[}CLS{]} The sentence ' {[}X{]} ' has the same meaning with {[}MASK{]}.{[}SEP{]} \\
Template 3 & {[}CLS{]} {[}MASK{]} has similar meaning with sentence : '{[}X{]}'.{[}SEP{]}      & {[}CLS{]} {[}MASK{]} has similar meaning with sentence : '{[}X{]}'.{[}SEP{]}      \\ \bottomrule
\end{tabular}
\end{adjustbox}
\label{sup:table:prompts}
\end{table}

\subsection{TC Minimization: Sentence Representation}

\paragraph{Training Setup}
We train the model with $10^6$ randomly sampled sentences from the English Wikipedia \citep{wikidump} dataset. 

\paragraph{Model Structure}
The model consists of two parts, $f$ is the transformer-based large pretrained language model and $g$ is a multiple layer perception with one hidden layer and ReLU activation function. The hidden neurons and output neurons are 256. The structure of the model follows \citet{gao2021simcse, DBLP:journals/corr/abs-2201-04337}.

\paragraph{Hyper-parameters}
Considering that \(f\) and \(g\) are initialized in different ways, we scale up the learning rate of \(g\). In detail, the learning rate and learning rate scale are $1\mathrm{e}{-5}$ and 10 for models with SimCSE-BERT, SimCSE-RoBERTa and PromptBERT as baselines, $5\mathrm{e}{-6}$ and 100 for the model with PromptRoBERTa as baselines. These hyper-parameters are selected according to the validation set. 
The batch size is 256. The maximum length of a sentence is limited to 32. The training epoch is one. The best model selected based on the performance of the validation set is applied to the testing set. The validation set is evaluated every 250 steps.

\paragraph{Augmentations}
Following the simple but effective data augmentation techniques in SimCSE, we get four augmentations by feeding the same input to the network four times and sample four dropout samples independently. 
Prompt-BERT also applies this dropout strategy, while they modify the inputs as well as the output feature representation neuron. 

SimCSE and PromptBERT both apply the InfoNCE with cosine similarity as the score function and use the large masked language pretrained model like BERT\citep{devlin2018bert} as the initialization of the feature encoder \(f(\cdot)\). 
The difference is how they construct the paired distribution \(p(x,y)\). 
SimCSE feeds the same sentence to the encoder twice and obtains different output embeddings by independently sampling two dropout masks of the encoder \(f(\cdot)\) ([CLS] representation is taken as the representation of the sentence). 
PromptBERT adapts two fixed cloze-style templates to augment the original sentence. 
For example, a template could be ``this sentence : [X] means [MASK].'', where [X] is the placeholder of the origin input sentence and [MASK] is a special placeholder token. The  representation of [MASK] token \(f(X)_{\text{[MASK]}}\) is donated as the sentence embedding.

In our total correlation estimation experiments, we augment the input sentence four times. More specifically, we sample four drop-outs for SimCSE and design four cloze-style templates for Prompt-BERT (~\Tableref{sup:table:prompts}). Our design of the template partially follows \citet{DBLP:journals/corr/abs-2201-04337}.


\end{document}